\documentclass[]{article}
\usepackage{proceed2e}

\usepackage{graphicx} 
\usepackage{subfigure}
\usepackage{wrapfig}

\usepackage{natbib}
\usepackage{multicol}
 
\usepackage{algorithm}
\usepackage{algorithmic} 
\usepackage{bm}
\usepackage{amsmath}
\usepackage{hyperref}
\usepackage{amssymb,amsmath,amsthm}
\usepackage{amsfonts}
\usepackage{latexsym}
\usepackage{multirow}

\usepackage{url}
\usepackage{color}
\usepackage{comment}
\usepackage{natbib}
\usepackage{tikz}
 \usepackage{picins}
\usepackage{hyperref}
\usepackage{algorithm,algorithmic,multirow,array,booktabs}
\usetikzlibrary{arrows,fit,positioning,shapes,snakes}

\let\oldbibliography\thebibliography
\renewcommand{\thebibliography}[1]{%
  \oldbibliography{#1}%
  \setlength{\itemsep}{1pt}%
}

\newdimen\arrowsize 
\pgfarrowsdeclare{arcsq}{arcsq}
{
  \arrowsize=0.2pt
  \advance\arrowsize by .5\pgflinewidth
  \pgfarrowsleftextend{-4\arrowsize-.5\pgflinewidth}
  \pgfarrowsrightextend{.5\pgflinewidth}
}
{
  \arrowsize=1.5pt
  \advance\arrowsize by .5\pgflinewidth
  \pgfsetdash{}{0pt} 
  \pgfsetroundjoin   
  \pgfsetroundcap    
  \pgfpathmoveto{\pgfpoint{0\arrowsize}{0\arrowsize}}
  \pgfpatharc{-90}{-140}{4\arrowsize}
  \pgfusepathqstroke
  \pgfpathmoveto{\pgfpointorigin}
  \pgfpatharc{90}{140}{4\arrowsize}
  \pgfusepathqstroke
}

\newcommand\independent{\protect\mathpalette{\protect\independenT}{\perp}}
\def\independenT#1#2{\mathrel{\rlap{$#1#2$}\mkern2mu{#1#2}}}

\newtheorem{Defi}{Definition}
\newtheorem{lemma}[Defi]{Lemma}
\newtheorem{theorem}[Defi]{Theorem}
\newtheorem{corollary}[Defi]{Corollary}
\newtheorem{proposition}[Defi]{Proposition}

\usepackage[utf8]{inputenc} 
\usepackage[T1]{fontenc}    
\usepackage{hyperref}       
\usepackage{url}            
\usepackage{booktabs}       
\usepackage{amsfonts}       
\usepackage{nicefrac}       
\usepackage{microtype}      

\title{Causal Discovery in the Presence of Measurement Error: Identifiability Conditions}

\author{{\small Kun Zhang$^\dagger$,~Mingming Gong$^{\star \dagger}$,~Joseph Ramsey$^\dagger$,~Kayhan Batmanghelich$^\star$,~Peter Spirtes$^\dagger$,~Clark Glymour$^\dagger$}\\ 
$^\dagger$Department of philosophy, Carnegie Mellon University\\
$^\star$Department of Biomedical Informatics, University of Pittsburgh
}

\begin{document}

\maketitle
\begin{abstract} Measurement error in the observed values of the variables can greatly change the output of various causal discovery methods. This problem has received much attention in multiple fields,  but it is not clear to what extent the causal model for the measurement-error-free variables can be identified in the presence of measurement error with unknown variance. In this paper, we study precise sufficient identifiability conditions for the measurement-error-free causal model and show what information of the causal model can be recovered from observed data. 
In particular, we present two different sets of identifiability conditions, based on the second-order statistics and higher-order statistics of the data, respectively. The former was inspired by the relationship between the generating model of the measurement-error-contaminated data and the factor analysis model, and the latter makes use of the identifiability result of the over-complete independent component analysis problem.

\end{abstract}





\section{Introduction}


Understanding and using  causal relations among variables of interest has been a fundamental problem in various fields, including biology, neuroscience, and social sciences. Since interventions or controlled randomized experiments are usually expensive or even impossible to conduct, discovering causal information from observational data, known as causal discovery~\citep{Spirtes00,Pearl00}, has been an important task and  received much attention in computer science, statistics, and philosophy. Roughly speaking, methods for causal discovery are categorized into  constraint-based ones, such as the PC algorithm~\citep{Spirtes00}, and score-based ones, such as Greedy Equivalence Search (GES)~\citep{chickering2002optimal}.

Causal discovery algorithms aim to find the causal relations among the observed variables. However, in many cases the measured variables are not identical to the variables we intend to measure. For instance, the measured brain signals may contain error introduced by the instruments, and in social sciences many variables are not directly measurable and one usually resorts to proxies (e.g., for ``regional security" in a particular area). In this paper, we assume that the observed variables $X_i$, $i=1,...,n$, are generated from the underlying measurement-noise-free variables  $\tilde{X}_i$  with additional random measurement errors ${E}_i$: 
\begin{equation} \label{Eq:contam}
X_i = \tilde{X}_i + E_i.
\end{equation}
Here we assume that the measurement errors $E_i$ are independent from $\tilde{X}_i$ and have non-zero variances. We call this model
the CAusal Model with Measurement Error (CAMME). 
Generally speaking, because of the presence of measurement errors, the d-separation patterns among $X_i$ are different from those among the underlying variables $\tilde{X}_i$. This generating process has been called the random measurement error model in ~\citep{Scheines16}. According to the causal Markov condition~\citep{Spirtes00,Pearl00}, observed variables $X_i$ and the underlying variables $\tilde{X}_i$ may have different conditional independence/dependence relations and, as a consequence, 
the output of constraint-based approaches to causal discovery is sensitive to such error, as demonstrated in~\citep{Scheines16}. Furthermore, because of the measurement error, the structural equation models according to which the measurement-error-free variables $\tilde{X}_i$ are generated usually  do not hold for the observed variables $X_i$. (In fact, $X_i$ follow error-in-variables models, for which the identifiability of the underlying causal relation is not clear.) Hence, approaches based on structural equation models, such as the linear, non-Gaussian, acyclic model (LiNGAM~\citep{Shimizu06}), 
will generally fail to find the correct causal direction and causal model.

In this paper, we aim to estimate the causal model underlying the measurement-error-free variables $\tilde{X}_i$ from their observed values $X_i$ contaminated by random measurement error. We assume linearity of the causal model and causal sufficiency relative to $\{\tilde{X}_i\}_{i=1}^n$. We particularly focus on the case where the causal structure for $\tilde{X}_i$ is represented by a Directed Acyclic Graph (DAG), although this condition can be weakened.  In order to develop principled causal discovery methods to recover the causal model for $\{\tilde{X}_i\}_{i=1}^n$ from observed values of $\{{X}_i\}_{i=1}^n$, we have to address theoretical issues include 
\begin{itemize}
\item whether the causal model of interest is completely or partially identifiable from the contaminated observations,
\item  what are the precise identifiability conditions, and
\item  what information in the measured data is essential for estimating the identifiable causal knowledge.
\end{itemize}
We make an attempt to answer the above questions on both theoretical and methodological sides.

One of the main difficulties in dealing with causal discovery in the presence of measurement error is because the variances of the measurement errors are unknown. Otherwise, if they are known, one can readily calculate the covariance matrix of the measurement-error-free variables $\tilde{X}_i$ and apply traditional causal discovery methods such as the PC ~\citep{Spirtes00} or GES~\citep{chickering2002optimal}) algorithm. 
It is worth noting that there exist causal discovery methods to deal with  confounders, i.e., hidden direct common causes, such as the Fast Causal Inference (FCI) algorithm~\citep{Spirtes00}. However, they cannot estimate the causal structure over the latent variables, which is what we aim to recover in this paper. \citep{Silva06} and \citep{Kummerfeld14} have provided algorithms for recovering latent variables and their causal relations when each latent variable has multiple measured effects.  Their problem is different from the measurement error setting we consider, where clustering for latent common causes is not required and each measured variable is the direct effect of a single "true" variable. Furthermore, as shown in next section, their models can be seen as special cases of our setting.


\section{Effect of Measurement Error on Conditional Independence / Dependence}

We use an example to demonstrate how measurement error changes the (conditional) independence and dependence relationships in the data. More precisely, we will see how the (conditional) independence and independence relations between the observed variables $X_i$ are different from those between the measurement-error-free variables $\tilde{X}_i$.  Suppose we observe $X_1$, $X_2$, and $X_3$, which are generated from measurement-error-free variables according to the structure given in Figure~\ref{fig:illust_C}. Clearly $\tilde{X}_1$ is dependent on $\tilde{X}_2$, while $\tilde{X}_1$ and $\tilde{X}_3$ are conditionally independent given $\tilde{X}_2$.
 One may consider general settings for the variances of the measurement errors. For simplicity, here let us assume that there is only measurement error in $X_2$, i.e., $X_1 = \tilde{X}_1$, $X_2 = \tilde{X}_2 + E_2$, and $X_3 = \tilde{X}_3$. 

\begin{figure}[htp]
\setlength{\abovecaptionskip}{0pt}
\setlength{\belowcaptionskip}{0.5pt}
  \begin{center}
\setlength{\abovecaptionskip}{-0.2pt}
\setlength{\belowcaptionskip}{0pt}
\begin{center}
\begin{tikzpicture}[scale=.7, line width=0.5pt, inner sep=0.2mm, shorten >=.1pt, shorten <=.1pt]
\draw (0, 0) node(1)  {{\footnotesize\,$\tilde{X}_1$\,}};
  \draw (1.8, 0) node(2) {{\footnotesize\,$\tilde{X}_2$\,}};
\draw (3.6, 0) node(3)  {{\footnotesize\,$\tilde{X}_3$\,}};
\draw (0, -1) node(4)  {{\footnotesize\,${X}_1$\,}};
  \draw (1.8, -1) node(5) {{\footnotesize\,${X}_2$\,}};
\draw (3.6, -1) node(6)  {{\footnotesize\,${X}_3$\,}};
  \draw[-arcsq] (2) -- (1); 
  \draw[-arcsq] (2) -- (3); 
  \draw[-arcsq] (1) -- (4); 
  \draw[-arcsq] (2) -- (5); 
  \draw[-arcsq] (3) -- (6); 
\end{tikzpicture} 
\end{center}
\caption{A linear CAMME to demonstrate the effect of measurement error on conditional independence and dependence relationships. For simplicity, we consider the special case where there is measurement error only
in $X_2$, i.e., $X_2 = \tilde{X}_2 + E_2$, but $X_1 = \tilde{X}_1$ and $X_3 = \tilde{X}_3$.}
\label{fig:illust_C}
  \end{center}
  \vspace{-10pt}
  \vspace{1pt}
\end{figure}
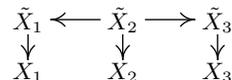

Let $\tilde{\rho}_{12}$ be the correlation coefficient between $\tilde{X_1}$ and $\tilde{X_2}$ and $\tilde{\rho}_{13,2}$ be the partial correlation coefficient between $\tilde{X_1}$ and $\tilde{X_3}$ given $\tilde{X}_2$, which is zero. Let ${\rho}_{12}$ and  ${\rho}_{13,2}$ be the corresponding correlation coefficient and partial correlation coefficient in the presence of measurement error. We also let $\tilde{\rho}_{12} = \tilde{\rho}_{23} = \tilde{\rho}$ to make the result simpler. So we have ${\rho}_{13} = \tilde{\rho}_{13} = \tilde{\rho}_{12} \tilde{\rho}_{23} = \tilde{\rho}^2$.  Let $\gamma = \frac{\mathbb{S}\textrm{td}(E_2)}{\mathbb{S}\textrm{td}(\tilde{X}_2)}$. For the data with measurement error,
\begin{flalign} \nonumber
{\rho}_{12} &=  \frac{\mathbb{C}\textrm{ov} (X_1,X_2)}{ \mathbb{V}\textrm{ar}^{1/2}(X_1) \mathbb{V}\textrm{ar}^{1/2}(X_2) } \\ \nonumber &= \frac{\mathbb{C}\textrm{ov} (\tilde{X}_1,\tilde{X}_2)}{ \mathbb{V}\textrm{ar}^{1/2}(\tilde{X}_1) (\mathbb{V}\textrm{ar}(\tilde{X}_2) + \mathbb{V}\textrm{ar}(E_2))^{1/2} } \\ \nonumber &= 
\frac{\tilde{\rho} }{(1 + \gamma^2)^{1/2}}; \\ \nonumber
 {\rho}_{13,2} &=  \frac{\rho_{13} - \rho_{12}\rho_{23} }{(1-\rho_{12}^2)^{1/2}(1-\rho_{23}^2)^{1/2}} \\ \nonumber &= \frac{ \tilde{\rho}_{13} - \frac{\tilde{\rho}_{12}\tilde{\rho}_{23} }{1+\gamma^2} } { \big( 1-\frac{\tilde{\rho}^2 }{(1 + \gamma^2)} \big)^{1/2}   \big(1-\frac{\tilde{\rho}^2 }{(1 + \gamma^2)} \big)^{1/2}    } 
\\ \nonumber &=  \frac{r^2 \tilde{\rho}^2}{1+\gamma^2 - \tilde{\rho}^2}.
\end{flalign}
As the variance of the measurement error in $X_2$ increases, $\gamma$ become larger, and $\rho_{12}$ decreases and finally goes to zero; in contrast, $ {\rho}_{13,2}$, which is zero for the measurement-error-free variables, is increasing and finally converges to $\tilde{\rho}^2$. See Figure \ref{fig:illustrate_rho}
for an illustration. In other words, in this example as the variance of the measurement error in $X_2$ increases, $X_1$ and $X_2$ become more and more independent, while $X_1$ and $X_3$ 
are conditionally more and more dependent given $X_2$. 
However, for the measurement-error-free variables, $\tilde{X}_1$ and $\tilde{X}_2$ are dependent
and $\tilde{X}_1$ and $\tilde{X}_3$ and conditionally independent given $\tilde{X}_2$. Hence, the structure given by constraint-based approaches to causal discovery on the observed variables can be very different from the causal structure over
measurement-error-free variables.


One might apply other types of methods instead of the constraint-based ones for causal discovery from data
with measurement error. In fact, as the measurement-error-free variables are not observable, $\tilde{X}_2$ in Figure \ref{fig:illust_C} is
actually a confounder for observed variables. As a consequence, generally speaking, due to the effect of the confounders, the independence noise assumption underlying functional causal model-based approaches,
such as the method based on the linear, non-Gaussian, acyclic model~\citep{Shimizu06}, will not hold for the observed variables any more. Figure~\ref{fig:illustrate_ind} gives an illustration on this. Figure~\ref{fig:illustrate_ind}(a) shows the scatter plot of $X_1$ vs. $X_2$
and the regression line from $X_2$ to $X_1$, where $\tilde{X}_2$, the noise in $\tilde{X}_1$, and the measurement error $E_2$, are all
uniformly distributed ($\rho = 0.4$, and $\gamma
 = 1.4$). As seen from Figure~\ref{fig:illustrate_ind}(b), the residual of regressing $X_1$ on $X_2$
is not independent from $X_2$, although the residual of regressing $\tilde{X}_1$ on $\tilde{X}_2$ is independent from $\tilde{X}_2$. As a result, the functional causal model-based approaches to causal discovery may also fail to find the causal structure of the measurement-error-free variables from their contaminated observations.

\begin{figure}[htp]
\centering
\includegraphics[width=6cm,height = 4.8cm,trim={4.2cm 9.3cm 5cm 9.5cm},clip]{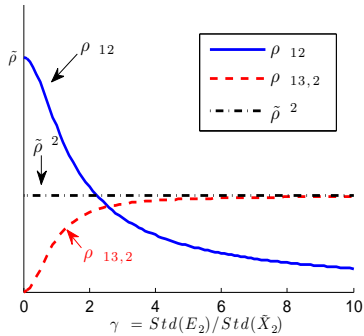}
 \caption{The correlation coefficient $\rho_{12}$ between $X_1$ and $X_2$ and partial correlation coefficient $\rho_{13,2}$ between $X_1$ and $X_3$ given $X_2$ as functions of $\gamma$, the ratio of the standard deviation of measurement error to the that of $\tilde{X}_2$. We have assumed that the correlation coefficient between $\tilde{X}_1$ and $\tilde{X}_2$ and that between $\tilde{X}_2$ and $\tilde{X}_3$ are the same (denoted by $\tilde{\rho}$), and that there is measurement error only in $X_2$.}
\label{fig:illustrate_rho}
\end{figure}
\begin{figure}[ht] 
\centering
\includegraphics[width=5.5cm,height = 4.3cm,trim={4.2cm 9.3cm 5cm 9.5cm},clip]{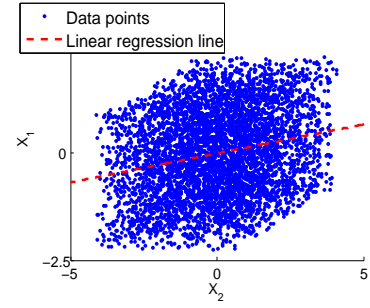}\\ (a) \\
\includegraphics[width=5.5cm,height = 4.3cm,trim={4.2cm 9.3cm 5cm 9.5cm},clip]{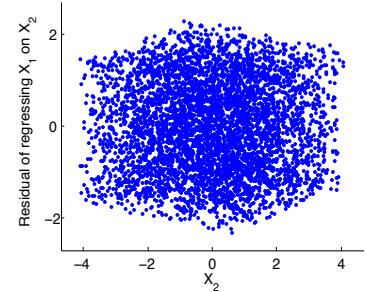}\\ (b) \\
 \caption{Illustration on how measurement error leads to dependence between regression residual and
contaminated cause. (a) Scatter plot of $X_2$ and $X_1$ with measurement error in $X_2$ together with the regression
line. (b) Scatter plot of the regression residual and $X_2$. Note that if we regress $\tilde{X}_1$ on $\tilde{X}_2$, the residual is
independent from $\tilde{X}_2$.}
\label{fig:illustrate_ind}
\end{figure}


\section{Canonical Representation of Causal Models with Measurement Error}

Let $\tilde{G}$ be the acyclic causal model over $\tilde{X}_i$. Here we call it {\it measurement-error-free causal model}.  
Let $\mathbf{B}$ be the corresponding causal adjacency matrix for $\tilde{X}_i$, in which $B_{ij}$ is the coefficient of the direct causal influence from $\tilde{X}_j$ to $\tilde{X}_i$ and $B_{ii} = 0$. We have,
\begin{equation}\label{Eq:B}
\tilde{\mathbf{X}} = \mathbf{B} \tilde{\mathbf{X}} + \tilde{\mathbf{E}},
\end{equation}
where the components of $\tilde{\mathbf{E}}$, $\tilde{E}_i$, have non-zero, finite variances.  
Then $\tilde{\mathbf{X}}$ is actually a linear transformation of the error terms in $\tilde{\mathbf{E}}$ because (\ref{Eq:B}) implies
\begin{equation}\label{Eq:A}
\tilde{\mathbf{X}} = \underbrace{(\mathbf{I} - \mathbf{B})^{-1}}_{\triangleq \mathbf{A}}\tilde{\mathbf{E}}.
\end{equation}

Now let us consider two types of nodes of $\tilde{G}$, namely, leaf nodes (i.e., those that do not influence any other node) and non-leaf nodes. Accordingly, the noise term in their structural equation models also has distinct behaviors: If $\tilde{X}_i$ is a leaf node, then $\tilde{E}_i$ influences only $\tilde{X}_i$, not any other; otherwise $\tilde{E}_i$ influences $\tilde{X}_i$ and at least one other variable, $\tilde{X}_j$, $j\neq i$. Consequently, we can decompose the noise vector into two groups: $\tilde{\mathbf{E}}^L$ consists of the $l$ noise terms that influence only leaf nodes, and $\tilde{\mathbf{E}}^{\mathrm{\mathrm{NL}}}$ contains the remaining noise terms.  Equation (\ref{Eq:A}) can be rewritten as
\begin{equation} \label{A2}
\tilde{\mathbf{X}} = \mathbf{A}^{\mathrm{\mathrm{NL}}} \tilde{\mathbf{E}}^{\mathrm{\mathrm{NL}}} + \mathbf{A}^{L} \tilde{\mathbf{E}}^{L} = \tilde{\mathbf{X}}^*  + \mathbf{A}^{L} \tilde{\mathbf{E}}^{L},
\end{equation}
where $\tilde{\mathbf{X}}^* \triangleq \mathbf{A}^{\mathrm{\mathrm{NL}}} \tilde{\mathbf{E}}^{\mathrm{\mathrm{NL}}}$, 
 $\mathbf{A}^{\mathrm{\mathrm{NL}}}$ and $\mathbf{A}^{L}$ are $n \times (n-l)$ and $n \times l$ matrices, respectively. Here both $A^L$ and $\mathbf{A}^{\mathrm{\mathrm{NL}}} $ have specific structures. All entries of $\mathbf{A}^L$ are 0 or 1; for each column of $\mathbf{A}^L$, there is only one non-zero entry. In contrast, each column of $\mathbf{A}^{\mathrm{\mathrm{NL}}}$ has at least two non-zero entries, representing the influences from the corresponding non-leaf noise term.

Further consider the generating process of observed variables $X_i$. Combining (\ref{Eq:contam}) and (\ref{A2}) gives
\begin{flalign}  \nonumber
\mathbf{X} &= \tilde{\mathbf{X}}^*  + \mathbf{A}^{L} \tilde{\mathbf{E}}^{L} + \mathbf{E}
=  \mathbf{A}^{\mathrm{\mathrm{NL}}} \tilde{\mathbf{E}}^{\mathrm{\mathrm{NL}}} + ( \mathbf{A}^{L} \tilde{\mathbf{E}}^{L} + \mathbf{E} ) \\ \label{Eq:FA}
&= \mathbf{A}^{\mathrm{\mathrm{NL}}} \tilde{\mathbf{E}}^{\mathrm{\mathrm{NL}}} + \mathbf{E}^* \\ \label{Eq:ICA}
& = \left[
    \begin{array}{c|c}
        \mathbf{A}^{\mathrm{\mathrm{NL}}} & \mathbf{I}    \end{array}
\right] \cdot \left[
    \begin{array}{c}
        \tilde{\mathbf{E}}^{\mathrm{\mathrm{NL}}}\\ \hline \vspace{-.27cm} \\
        \mathbf{E}^*
    \end{array}
\right],
\end{flalign}
where $\mathbf{E}^* = \mathbf{A}^{L} \tilde{\mathbf{E}}^{L} + \mathbf{E}$ and $\mathbf{I}$ denotes the identity matrix.  To make it more explicit, we give how $X_i^*$ and $E_i^*$ are related to the original CAMME process:
\begin{flalign} \label{newE}
\tilde{X}_i^* &= \begin{cases}
                        \tilde{X}_i, &\text{if $\tilde{X}_i$ is not a leaf node in $\tilde{G}$}; \\
                        \tilde{X}_i - \tilde{E}_i, &\text{otherwise;}
                    \end{cases},~\textrm{and}~ 
\\ \nonumber E_i^* &= \begin{cases}
                        E_i, &\text{if $\tilde{X}_i$ is not a leaf node in $\tilde{G}$}; \\
                        E_i + \tilde{E}_i, &\text{otherwise.}
                    \end{cases}
\end{flalign}
Clearly $E^*_i$s are independent across $i$, and as we shall see in Section~\ref{Sec:ident}, the information shared by difference $X_i$ is still captured by $\tilde{\mathbf{X}}^*$.

\begin{proposition}
For each \textcolor{black}{CAMME} specified by (\ref{Eq:B}) and (\ref{Eq:contam}), there always exists an observationally equivalent representation in the form of (\ref{Eq:FA}) or (\ref{Eq:ICA}),
\end{proposition}
The proof was actually given in the construction procedure of the representation (\ref{Eq:FA}) or (\ref{Eq:ICA}) from the original CAMME. 
We call the representation (\ref{Eq:FA}) or (\ref{Eq:ICA}) the canonical representation of the underlying CAMME (CR-CAMME). 

\paragraph{Example Set 1}
Consider the following example with three observed variables $X_i$, $i=1,2,3$, for which $\tilde{X}_1 \rightarrow \tilde{X}_2 \leftarrow \tilde{X}_3$, with causal relations $\tilde{X}_2 = a \tilde{X}_1 + b\tilde{X}_3 +\tilde{E}_2$.  That is, 
$$\mathbf{B} = \begin{bmatrix} 0 & 0 & 0 \\  a & 0 & b \\ 0 & 0 &0 
\end{bmatrix}, $$ 
and according to (\ref{Eq:A}), 
$$\mathbf{A} = \begin{bmatrix} 1 & 0 & 0 \\  a & 1 & b \\ 0 & 0 &1 
\end{bmatrix}.$$ 
Therefore, 
\begin{flalign} \nonumber
\mathbf{X} &= \tilde{\mathbf{X}} + \mathbf{E} = \tilde{\mathbf{X}}^* + \mathbf{E}^*
\\ \nonumber
&= \begin{bmatrix} 1  & 0 \\  a  & b \\ 0  &1 
\end{bmatrix} \cdot \begin{bmatrix} \tilde{E}_1 \\ \tilde{E}_3
\end{bmatrix}  + \begin{bmatrix} E_1 \\  \tilde{E}_2 + E_2 \\ E_3 
\end{bmatrix} \\ \nonumber
&= \begin{bmatrix} 1  & 0 &\vline & 1 & 0 &0 \\  a  & b &\vline & 0 & 1 &0 \\ 0  &1 &\vline & 0 & 0 & 1 
\end{bmatrix} \cdot \begin{bmatrix} \tilde{E}_1 \\ \tilde{E}_3
\\ E_1 \\  \tilde{E}_2 + E_2 \\ E_3 
\end{bmatrix}.
\end{flalign}


In causal discovery from observations in the presence of measurement error, we aim to recover information of the measurement-error-free causal model $\tilde{G}$. 
Let us define a new graphical model, $\tilde{G}^*$. It is obtained by replacing variables $\tilde{X}_i$ in $\tilde{G}$ with variables $\tilde{X}_i^*$. In other words, it has the same causal structure and causal parameters (given by the $\mathbf{B}$ matrix) as $\tilde{G}$, but its nodes correspond to variables $\tilde{X}^*_i$. If we manage to estimate the structure of and involved causal parameters in $\tilde{G}^*$, then $\tilde{G}$, the causal model of interest, is recovered. Comparing with $\tilde{G}$,  $\tilde{G}^*$ involves some deterministic causal relations because each leaf node is a deterministic function of its parents (the noise in leaf nodes has been removed; see (\ref{newE})). We defined the graphical model $\tilde{G}^*$ because we cannot fully estimate the distribution of measurement-error-free variables $\tilde{\mathbf{X}}$, but might be able to estimate that of $\tilde{\mathbf{X}}^*$, under proper assumptions.

In what follows, most of the time we assume
\begin{itemize}
\item[A0.] The causal Markov condition holds for $\tilde{G}$ and the distribution of $\tilde{X}_i^*$ is {\it non-deterministically faithful} w.r.t. $\tilde{G}^*$, in the sense that if there exists $\mathbf{S}$, a subset of $\{\tilde{X}_k^*\,:\, k\neq i, k\neq j\}$, such that neither of $\tilde{X}_i^*$ and $\tilde{X}_j^*$ is a deterministic function of $\mathbf{S}$ and $\tilde{X}_i^* \independent \tilde{X}_j^* \,|\, \mathbf{S}$ holds, then $\tilde{X}_i^*$ and $\tilde{X}_j^*$ (or $\tilde{X}_i$ and $\tilde{X}_j$) are d-separated by $\mathbf{S}$ in $\tilde{G}^*$.
\end{itemize}
This non-deterministically faithfulness assumption excludes a particular type of parameter coupling in the causal model for $\tilde{X}_i$. in Figure~\ref{fig:faithful} we give a causal model in which the causal coefficients are carefully chosen so that this assumption is violated: because $\tilde{X}^*_3 = a \tilde{X}^*_1 + b \tilde{X}^*_2$ and $\tilde{X}^*_4 = 2a \tilde{X}^*_1 + 2b \tilde{X}^*_2 + E_4^*$, we have $\tilde{X}^*_4 = 2\tilde{X}^*_3 +  E_4^*$, implying $\tilde{X}_4^* \independent \tilde{X}_1^* \,|\, \tilde{X}_3^*$ and $\tilde{X}_4^* \independent \tilde{X}_2^* \,|\, \tilde{X}_3^*$, which are not given by the causal Markov condition on $\tilde{G}$. We note that this non-deterministic faithfulness is defined for the distribution of the constructed variables $\tilde{X}_i^*$, not the measurement-error-free variables $\tilde{X}_i$. (Bear in mind their relationship given in (\ref{newE}).) This assumption is generally stronger than the faithfulness assumption for the distribution of $\tilde{X}_i$. In particular, in the causal model given in Figure~\ref{fig:faithful}, the distribution of $\tilde{X}_i$ is still faithful w.r.t. $\tilde{G}$. Below we call 
the conditional independence relationship between $\tilde{X}_i^*$ and $\tilde{X}_j^*$ given $\mathbf{S}$  where 
neither of $\tilde{X}_i^*$ and $\tilde{X}_j^*$ is a deterministic function of $\mathbf{S}$ {\it non-deterministic conditional independence}.
 
\begin{figure}[h]
\centering
{\hspace{0cm}\begin{tikzpicture}[scale=.7, line width=0.5pt, inner sep=0.2mm, shorten >=.1pt, shorten <=.1pt]
\draw (0, 0) node(2)  {{\footnotesize\,$\tilde{X}_4$\,}};
  \draw (-2.6, 0) node(1) {{\footnotesize\,$\tilde{X}_2$\,}};
\draw (2.6, 0) node(3)  {{\footnotesize\,$\tilde{X}_5$\,}};
\draw (-5.2, 0) node(5) {{\footnotesize\,$\tilde{X}_1$\,}};
\draw (-3.9, -1.5) node(6) {{\footnotesize\,$\tilde{X}_3$\,}};
  \draw[-arcsq] (1) -- (2); \node[above] at (-1.3,0) {{$2b$}};
  \draw[-arcsq] (2) -- (3); \node[above] at (1.3,0) {{$d$}};
        \draw[-arcsq] (5) -- (1); \node[above] at (-3.9,0) {{$c$}};
            \draw[-arcsq] (5) -- (6); \node[above] at (-4.5,-0.7) {{$a$}};
                        \draw[-arcsq] (1) -- (6); \node[above] at (-3.3,-.7) {{$b$}};
                        \draw [-arcsq] (5) to[bend left=25]  (2); \node[above] at (-2.6,0.79) {{$2a$}};
\end{tikzpicture}} 
\caption{A causal model in which $\tilde{X}^*_i$ are not {\it non-deterministically faithful} w.r.t. $\tilde{G}$ because of parameter coupling.}
\label{fig:faithful}
\end{figure}
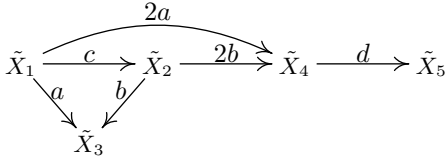

Now we have two concerns. One is whether essential information of the CR-CAMME is identifiable from observed values of $\mathbf{X}$.  We are interested in finding the causal model for (or a particular type of dependence structures in) $\tilde{\mathbf{X}}$. The CR-CAMME of $\mathbf{X}$, given by (\ref{Eq:FA}) or (\ref{Eq:ICA}), has two terms, $\tilde{\mathbf{X}}^*$ and $\tilde{\mathbf{E}}^*$.  The latter is independent across all variables, and the former preserves major information of the dependence structure in $\tilde{\mathbf{X}}$. Such essential information of the CR-CAMME may be the covariance matrix of $\tilde{\mathbf{X}}^*$ or the matrix $\mathbf{A}^{\mathrm{\mathrm{NL}}}$, as discussed in next sections. 
In the extreme case, suppose such information is not identifiable at all, then it is hopeless to find the underlying causal structure of $\tilde{G}$. The other is what information of the original CAMME, in particular, the causal model over the measurement-error-free variables, can be estimated from the above identifiable information of the CR-CAMME. Although the transformation from the original CAMME to a CR-CAMME is straightforward, without further knowledge there does not necessarily exist a unique CAMME corresponding to a given CR-CAMME: first, the CR-CAMME does not tell us which nodes $\tilde{X}_i$ are leaf nodes in $\tilde{G}$; second, even if $\tilde{X}_i$ is known to be a leaf node, it is impossible to separate the measurement error $E_i$ from the noise $\tilde{E}_i$ in $E^*_i$. Fortunately, we are not interested in everything of the original CAMME, but only the causal graph $\tilde{G}$ and the corresponding causal influences $\mathbf{B}$.

Accordingly, in the next sections we will explore what information of the CR-CAMME is identifiable from the observations of $\mathbf{X}$ and how to further reconstruct necessary information of the original CAMME.  
%
%
%
In the measurement error model (\ref{Eq:contam}) we assumed that each observed variable $X_i$ is generated from its own latent variable $\tilde{X}_i$. We note that in case multiple observed variables are generated from a single latent variable or a single observed variable is generated by multiple latent variables (see, e.g.,~\citep{Silva06}), we can still use the CR-CAMME to represent the process. In the former case, certain rows of $\mathbf{A}^{\mathrm{\mathrm{NL}}}$ are identical. For instance, if $X_1$ and $X_2$ are generated as noisy observations of the same latent variable, then in (\ref{Eq:FA}) the first two rows of $\mathbf{A}_{\mathrm{\mathrm{NL}}}$ are identical. (More generally, if one allows different coefficients to generate them from the latent variable, the two rows are proportional to each other.) Then let us consider an example in the latter case. Suppose $X_3$ is generated by latent variables $\tilde{X}_1$ and $\tilde{X}_2$, for each of which there is also an observable counterpart. Write the causal model as $X_3 = f(\tilde{X}_1, \tilde{X}_2) + E_3$ and introduce the latent variable $\tilde{X}_3 = f(\tilde{X}_1, \tilde{X}_2)$, and then we have $X_3 = \tilde{X}_3 + E_3$.  The CR-CAMME formulation then follows.

\section{Identifiability with Second Order Statistics} \label{Sec:ident}


The CR-CAMME (\ref{Eq:FA}) has a form of the factor analysis model (FA)~\citep{Everitt84}, which has been a fundamental tool in data analysis.  In its general form, FA assumes the observable random vector $\mathbf{X} = (X_1,X_2,...,X_n)^\intercal$ was generated by 
\begin{equation} \label{Eq:FA_model}
\mathbf{X} = \mathbf{Lf} + \mathbf{N},
\end{equation}
where the factors $\mathbf{f} = (f_1,...,f_r)^\intercal$ satisfies $\mathbb{C}\textrm{ov}(\mathbf{f}) = \mathbf{I}$, and noise terms, as components of $\mathbf{N}$, are mutually independent and also independent from $\mathbf{f}$. Denote by $\boldsymbol{\Psi}_N$ the covariance matrix of $\mathbf{N}$, which is diagonal.  The unknowns in (\ref{Eq:FA_model}) are the loading matrix $\mathbf{L}$ and the covariance matrix $\boldsymbol{\Psi}_N$. 

Factor analysis only exploits the second-order statistics, i.e., it assumes that all variables are jointly Gaussian. Clearly in FA $\mathbf{L}$ is not identifiable; it suffers from at least the right orthogonal transformation indeterminacy. However, under suitable conditions, some essential information of FA is generically identifiable, as given in the following lemma.


\begin{lemma} For the factor analysis model, when the number of factors $r < \phi(n) = \frac{2n+1 - (8n+1)^{1/2}}{2}$, the model is \textcolor{black}{generically} globally identifiable, in the sense that for randomly generated $(L, \boldsymbol{\Psi}_N)$ in (\ref{Eq:FA_model}), it is with only measure 0 that there exists another representation $(L', \boldsymbol{\Psi}'_N)$ such that $(L', \boldsymbol{\Psi}'_N)$ and $(L, \boldsymbol{\Psi}_N)$ generate the same covariance matrix for $\mathbf{X}$ and $\boldsymbol{\Psi}'_N \neq \boldsymbol{\Psi}_N$.
\end{lemma}
This was formulated as a conjecture by~\citep{Shapiro85}, and was later proven by~\citep{Bekker97}. This lemma immediately gives rise to the following {\it generic identifiability} of the variances of measurement errors.\footnote{\textcolor{black}{We note that this ``generic identifiability" is sightly weaker than what we want: we want to show that for certain $(\mathbf{L}, \boldsymbol{\Psi}_\mathbf{N})$ the model is necessarily identifiable. To give this proof is non-trivial and is a line of our future research.}}


\begin{proposition} \label{iden_2nd}
The variances of error terms $E_i^*$ and the covariance matrix of $\tilde{\mathbf{X}}^*$ in the CR-CAMME (\ref{Eq:FA}) are \textcolor{black}{generically} identifiable when the sample size $N\rightarrow \infty$ and the following assumption on the number of leaf nodes $l$ holds:
\begin{itemize}
\item[A1.] The number of leaf variables $l$ satisfies 
\begin{equation} \label{Eq:cn}
\frac{l}{n} > c(n) \triangleq \frac{(8n+1)^{1/2} -1 }{2n}.
\end{equation}
\end{itemize}
\end{proposition} 

Clearly $c(n)$ is decreasing in $n$ and $c(n) \rightarrow 0$ as $n \rightarrow \infty$.   To give a sense how restrictive the above condition is,  Fig.~\ref{fig:illustrate_cn} shows how $c(n)$ changes with $n$.  In particular, when $n=4$, $c(n) = 59.3\%$, condition (\ref{Eq:cn}) implies the number of leaf nodes is $l > 2.4$; when $n=100$, $c(n) = 13.6\%$, condition (\ref{Eq:cn}) implies $l > 13.6$. Roughly speaking, as $n$ increases, it is more likely for condition (\ref{Eq:cn}) to hold. Note that the condition given in Proposition \ref{iden_2nd} is sufficient but not necessary for the identifiability of the noise variances and the covariance matrix of the non-leaf hidden variables~\citep{Bekker97}. 

\begin{figure}[ht]
\centering
\includegraphics[width=5cm,height = 4cm,trim={4.2cm 9.3cm 5cm 9.5cm},clip]{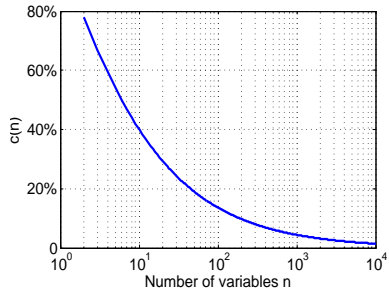}
 \caption{$c(n)$ as a function of $n$.}
\label{fig:illustrate_cn}
\end{figure}

Now we know that under certain conditions, the covariance matrices of $\mathbf{E}^*$ and $\tilde{\mathbf{X}}^*$ in the CR-CAMME (\ref{Eq:FA}) are (asymptotically) identifiable from observed data with measurement error. Can we recover the measurement-error-free causal model $\tilde{G}$ from them?


\subsection{Gaussian CAMME with the Same Variance For Measurement Errors}  \label{Sec:G_CAMME_same}
In many problems the variances of the measurement errors in different variables are roughly the same because the same instrument is used and the variables are measured in similar ways. For instance, this might approximately be the case for Functional magnetic resonance imaging (fMRI) recordings. In fact, if we made the following assumption on the measurement error, the underlying causal graph $\tilde{G}$ can be estimated at least up to the equivalence class, as shown in the following corollary.
\begin{itemize}
\item[A2.] The measurement errors in all observed variables have the same variance.
\end{itemize}

\begin{proposition} \label{coro1}
Suppose assumptions A0, A1, and A2 hold. Then as $N \rightarrow \infty$, $\tilde{G}$ can be estimated up to the equivalence class and, moreover, the leaf nodes of $\tilde{G}$ are identifiable. 
\end{proposition}

Proofs are given in Appendix. The proof of this corollary inspires a procedure to estimate the information of $\tilde{G}$ from contaminated observations in this case, which is denoted by $\texttt{FA+EquVar}$. 
 It consists of four steps. (1) Apply FA on the data with a given number of leaf nodes and estimate the variances of $E^*_i$ as well as the covariance matrix of $\tilde{\mathbf{X}}^*$.\footnote{Here we suppose the number of leaf nodes is given. In practice one may use model selection methods, such as BIC, to find this number.} (2) The $(n-l)$ smallest values of  the variances of $E^*_i$ correspond to non-leaf nodes, and the remaining $l$ nodes correspond to leaf nodes. (3) Apply a causal discovery method, such as the PC algorithm, to the sub-matrix of the estimated covariance matrix of $\tilde{\mathbf{X}}^*$ corresponding to non-leaf nodes and find the causal structure over non-leaf nodes. (4) For each leaf node $X_i^*$, find the subset of non-leaf nodes that determines $X_i^*$, and draw directed edges from those nodes to $X_i^*$, and further perform orientation propagation. 






\subsection{Gaussian CAMME: General Case} \label{Sec:G_CAMME_general} 
Now let us consider the general case where we do not have the constraint A2 on the measurement error. Generally speaking, after performing FA on the data, the task is to discover causal relations among $\tilde{X}_i^*$ by analyzing their estimated  covariance matrix, which is, unfortunately, singular, with the rank $(n - l)$. Then there must exist deterministic relations among $\tilde{{X}}^*_i$, and we have to deal with such relations in causal discovery.
Here suppose we simply apply the Deterministic PC (DPC) algorithm~\citep{Glymour07,Luo06} to tackle this problem. DPC is almost identical to PC, and the only difference is that when testing for conditional independence relationship $U\independent V\,|\,\mathbf{W}$, if  $U$ or $V$ is a deterministic function of $\mathbf{W}$, one then ignores this test (or equivalently we do not remove the edge between $U$ and $V$). 
We denote by $\texttt{FA+DPC}$ this procedure for causal discovery from data with measurement error. 

Under some conditions on the underlying causal model $\tilde{G}$, it can be estimated up to its equivalence class, as given in the following proposition. Here we use $\mathrm{PA}(\tilde{X}_j)$ to denote the set of parents (direct causes) of $\tilde{X}_j$ in $\tilde{G}$. 

\begin{proposition} \label{Prop:equivalence}
Suppose Assumptions A0 and A1 hold. As $N \rightarrow \infty$, compared to $\tilde{G}$, the graph produced by the above DPC procedure does not contain any missing edge. In particular, the edges between all non-leaf nodes are corrected identified. Furthermore, the whole graph of $\tilde{G}$ is identifiable up to its equivalence class if the following  assumption further holds:
\begin{itemize}
\item[A3.] For each pair of leaf nodes $\tilde{X}_j$ and $\tilde{X}_k$, there exists 
$\tilde{X}_p \in \mathrm{PA}(\tilde{X}_j)$
and 
$\tilde{X}_q \in \mathrm{PA}(\tilde{X}_k)$
that are d-separated in $\tilde{G}$ by a variable set $\mathbf{S}_{1}$, which may be the empty set. Moreover,  for each leaf node $\tilde{X}_j$ and each non-leaf node $\tilde{X}_i$ which are not adjacent, there exists 
$\tilde{X}_r \in \mathrm{PA}(\tilde{X}_j)$ which is d-separated from $\tilde{X}_i$ in $\tilde{G}$  by a variable set $\mathbf{S}_2$, which may be the empty set.
\end{itemize}

\end{proposition} 

\paragraph{Example Set 2 and Discussion}
Suppose assumption A0 holds. 
\begin{itemize}
\item $\tilde{G}_A$, given in Figure~\ref{fig:illust_D}(a), follows assumptions A1 and A3.  According to Proposition \ref{Prop:equivalence}, the equivalence class of this causal DAG can be asymptotically estimated from observations with measurement error.  
\item 
Assumptions A0, A1, and A3 are sufficient conditions for $\tilde{G}$ to be recovered up to its equivalence class and, they, especially A3, may not be necessary.  For instance, consider the causal graph $\tilde{G}_B$ given in Figure~\ref{fig:illust_D}(b), for which assumption A3 does not hold. If assumption A2 holds, $\tilde{G}_B$ can be uniquely estimated from contaminated data. Other constraints may also guarantee the identifiability of the underlying graph. For example, suppose all coefficients in the causal model are smaller than one in absolute value, then $\tilde{G}_B$  can also be uniquely estimated from noisy data. Relaxation of assumption A3 which still guarantees that $\tilde{G}$ is identifiable up to its equivalence class is a future line of research. 
\item The causal graphs $\tilde{G}_C$ and $\tilde{G}_D$, shown in Figure~\ref{fig:illust_D}(c), do not follow A1, so generally speaking, they are not identifiable from contaminated observations with second-order statistics. This is also the case for $\tilde{G}_E$, shown in Figure~\ref{fig:illust_D}(d).
\end{itemize}
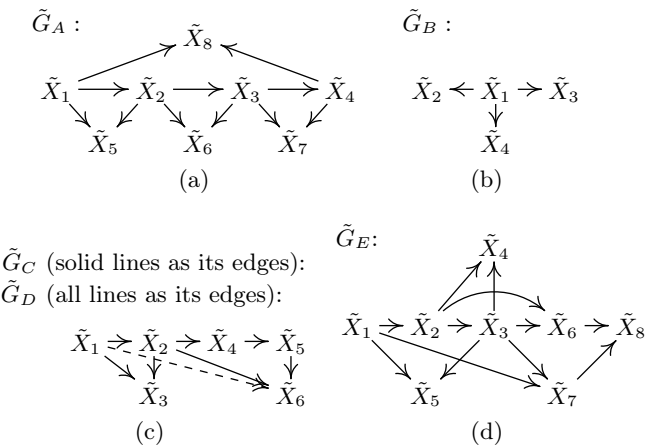
\begin{figure}[h]
\centering
\subfigure []
{\hspace{0cm}\begin{tikzpicture}[scale=.7, line width=0.5pt, inner sep=0.2mm, shorten >=.1pt, shorten <=.1pt]
\draw (0, 0) node(1)  {{\footnotesize\,$\tilde{X}_1$\,}};
  \draw (1.8, 0) node(2) {{\footnotesize\,$\tilde{X}_2$\,}};
\draw (3.6, 0) node(3)  {{\footnotesize\,$\tilde{X}_3$\,}};
\draw (5.4, 0) node(4) {{\footnotesize\,$\tilde{X}_4$\,}};
\draw (0.9, -1) node(5) {{\footnotesize\,$\tilde{X}_5$\,}};
  \draw (2.7, -1) node(6)  {{\footnotesize\,$\tilde{X}_6$\,}};
\draw (4.5, -1) node(7){{\footnotesize\,$\tilde{X}_7$\,}};
\draw (2.7, 1) node(8) {{\footnotesize\,$\tilde{X}_8$\,}};
\draw (0,1.3) node(9) {{\footnotesize\,$\tilde{G}_A:$\,}};
  \draw[-arcsq] (1) -- (2); 
  \draw[-arcsq] (2) -- (3); 
  \draw[-arcsq] (3) -- (4); 
  \draw[-arcsq] (1) -- (5); 
  \draw[-arcsq] (2) -- (5); 
    \draw[-arcsq] (2) -- (6); 
  \draw[-arcsq] (3) -- (6); 
  \draw[-arcsq] (3) -- (7); 
  \draw[-arcsq] (4) -- (7); 
  \draw[-arcsq] (1) -- (8); 
    \draw[-arcsq] (4) -- (8); 
\end{tikzpicture} }~~~
\subfigure []
{\hspace{0cm}\begin{tikzpicture}[scale=.7, line width=0.5pt, inner sep=0.2mm, shorten >=.1pt, shorten <=.1pt]
\draw (0, 0) node(1)  {{\footnotesize\,$\tilde{X}_1$\,}};
  \draw (-1.3, 0) node(2) {{\footnotesize\,$\tilde{X}_2$\,}};
\draw (1.3, 0) node(3)  {{\footnotesize\,$\tilde{X}_3$\,}};
\draw (0, -1) node(4) {{\footnotesize\,$\tilde{X}_4$\,}};
\draw (-1.3,1.3) node(5) {{\footnotesize\,$\tilde{G}_B:$\,}};
  \draw[-arcsq] (1) -- (2); 
  \draw[-arcsq] (1) -- (3); 
  \draw[-arcsq] (1) -- (4); 
\end{tikzpicture}}
\\ 
\subfigure []
{\hspace{0cm}\begin{tikzpicture}[scale=.7, line width=0.5pt, inner sep=0.2mm, shorten >=.1pt, shorten <=.1pt]
\draw (0, 0) node(2)  {{\footnotesize\,$\tilde{X}_4$\,}};
  \draw (-1.3, 0) node(1) {{\footnotesize\,$\tilde{X}_2$\,}};
\draw (1.3, 0) node(3)  {{\footnotesize\,$\tilde{X}_5$\,}};
\draw (1.3, -1) node(4) {{\footnotesize\,$\tilde{X}_6$\,}};
\draw (-2.6, 0) node(5) {{\footnotesize\,$\tilde{X}_1$\,}};
\draw (-1.3, -1) node(6) {{\footnotesize\,$\tilde{X}_3$\,}};
\draw (-1.3,1.5) node(7) {{\footnotesize\,$\tilde{G}_C$ (solid lines as its edges):\,}};
\draw (-1.3,.9) node(8) {{\footnotesize\,$\tilde{G}_D$  (all lines as its edges):~~~\,}};
  \draw[-arcsq] (1) -- (2); 
  \draw[-arcsq] (2) -- (3); 
  \draw[-arcsq] (3) -- (4); 
    \draw[-arcsq] (1) -- (4); 
        \draw[-arcsq] (5) -- (1); 
            \draw[-arcsq] (5) -- (6); 
                        \draw[-arcsq] (1) -- (6); 
                        \draw[-arcsq,dashed] (5) -- (4); 
\end{tikzpicture}}~~
\subfigure []
{\hspace{0cm}\begin{tikzpicture}[scale=.7, line width=0.5pt, inner sep=0.2mm, shorten >=.1pt, shorten <=.1pt]
\draw (0, 0) node(2)  {{\footnotesize\,$\tilde{X}_3$\,}};
  \draw (-1.3, 0) node(1) {{\footnotesize\,$\tilde{X}_2$\,}};
\draw (1.3, 0) node(3)  {{\footnotesize\,$\tilde{X}_6$\,}};
\draw (1.3, -1.3) node(4) {{\footnotesize\,$\tilde{X}_7$\,}};
\draw (2.6, 0) node(8)  {{\footnotesize\,$\tilde{X}_8$\,}};
\draw (-2.6, 0) node(5) {{\footnotesize\,$\tilde{X}_1$\,}};
\draw (-1.3, -1.3) node(6) {{\footnotesize\,$\tilde{X}_5$\,}};
\draw (0, 1.5) node(7) {{\footnotesize\,$\tilde{X}_4$\,}};
\draw (-2.6,1.7) node(9) {{\footnotesize\,$\tilde{G}_E$:\,}};
        \draw[-arcsq] (5) -- (1); 
  \draw[-arcsq] (1) -- (2); 
  \draw[-arcsq] (2) -- (3); 
  \draw[-arcsq] (3) -- (8); 
  \draw[-arcsq] (1) -- (7); 
    \draw[-arcsq] (2) -- (7); 
        \draw[-arcsq] (2) -- (4); 
  \draw[-arcsq] (5) -- (6); 
    \draw[-arcsq] (2) -- (6); 
  \draw[-arcsq] (4) -- (8); 
    \draw[-arcsq] (5) -- (4); 
\draw [-arcsq] (1) to[bend left=40]  (3);
\end{tikzpicture}}
\caption{(a)  $\tilde{G}_A$: a causal DAG $\tilde{G}$ which follows assumptions A1 and A3. (b) $\tilde{G}_B$: a DAG which follows assumption A1, but not A3; however, the structure is still identifiable if either assumption A2 holds or we know that all causal coefficients are smaller than one in absolute value. (c) Two DAGs that do not follow assumption A1; $\tilde{G}_C$ has only the solid lines as its edges, and $\tilde{G}_D$ also includes the dashed line. (d) $\tilde{G}_E$: another DAG that does not follow assumption A1.}
 \label{fig:illust_D}
\end{figure}

\section{Identifiability with Higher Order Statistics}\label{Subsec:higher_order}

The method based on second-order statistics exploits FA and deterministic causal discovery, both of which are computationally relatively efficient. However, if the number of leaf-nodes is so small that the condition in Proposition \ref{iden_2nd} is violated (roughly speaking, usually this does not happen when $n$ is big, say, bigger than 50, but is likely to be the case when $n$ is very small, say, smaller than 10), the underlying causal model is not guaranteed to be identifiable from contaminated observations. Another issue is that with second-order statistics, the causal model for $\tilde{\mathbf{X}}$ is usually not uniquely identifiable; in the best case it can be recovered up to its equivalence class (and leaf nodes). To tackle these issues, below we show that we can benefit from higher-order statistics of the noise terms.

In this section we further make the following assumption on the distribution of $\tilde{E}_i$:
\begin{itemize}
\item[A4.] All $\tilde{E}_i$ are non-Gaussian.
\end{itemize}

We note that under the above assumption, $\mathbf{A}^{\mathrm{\mathrm{NL}}}$ in (\ref{Eq:ICA}) can be estimated up to the permutation and scaling indeterminacies (including the sign indeterminacy) of the columns, as given in the following lemma.
\begin{lemma} \label{Lemma1_column_L}
Suppose assumption A4 holds. Given ${\mathbf{X}}$ which is generated according to (\ref{Eq:ICA}), $\mathbf{A}^{\mathrm{\mathrm{NL}}}$ is identifiable up to permutation and scaling of columns as the sample size $N \rightarrow \infty$.
\end{lemma}
\begin{proof}
This lemma is implied by Theorem 10.3.1 in \citep{Kagan73} or Theorem 1 in~\citep{Eriksson04}.
\end{proof}



\subsection{Non-Gaussian CAMME with the Same Variance For Measurement Errors}
We first note that under certain assumptions the underlying graph $\tilde{G}$ is fully identifiable, as shown in the following proposition.
\begin{proposition}\label{prop:recover_ica}
Suppose the assumptions in Corollary~\ref{coro1} hold, and further suppose assumption A4 holds. Then as $N \rightarrow \infty$, the underlying causal graph $\tilde{G}$ is fully identifiable from observed values of $X_i$.
\end{proposition}

\subsection{Non-Gaussian CAMME: More General Cases} \label{Sec:NG_general}
In the general case, what information of the causal structure $\tilde{G}$ can we recover? Can we apply existing methods for causal discovery based on LiNGAM, such as ICA-LiNGAM~\citep{Shimizu06} and Direct-LiNGAM~\citep{DirectLingam}, to recover it? LiNGAM assumes that the system is non-deterministic: each variable is generated as a linear combination of its direct causes plus a non-degenerate noise term. As a consequence, the linear transformation from the vector of observed variables to the vector of independent noise terms is a square matrix; ICA-LiNGAM applies certain operations to this matrix to find the causal model, and Direct-LiNGAM estimates the causal ordering by enforcing the property that the residual of regressing the effect on the root cause is always independent from the root cause.


In our case, $\mathbf{A}^{\mathrm{NL}}$, the essential part of the mixing matrix in (\ref{Eq:ICA}), is $n\times r$, where $r<n$. In other words, for some of the variables $\tilde{X}_i^*$, the causal relations are deterministic. (In fact, if $\tilde{X}_k$ is a leaf node in $\tilde{G}$, $\tilde{X}_k^*$ is a deterministic function of $\tilde{X}_k$'s direct causes.) As a consequence, unfortunately, the above causal analysis methods based on LiNGAM, including ICA-LiNGAM and Direct-LiNGAM, do not apply.  We will see how to recover information of $\tilde{G}$ by analyzing the estimated $\mathbf{A}^{\mathrm{NL}}$.


 We will show that some group structure and the group-wise causal ordering in $\tilde{G}$ can always be recovered. Before presenting the results, let us define the following recursive group decomposition according to causal structure $\tilde{G}$.  
\begin{Defi} [{\bf Recursive group decomposition}]
Consider the causal model $\tilde{G}^*$. Put all leaf nodes which share the same direct-and-only-direct node in the same group; further incorporate the corresponding direct-and-only-direct node in the same group. Here we say a node $\tilde{X}_i^*$ is the ``direct-and-only-direct" node of $\tilde{X}_j^*$ if and only if $\tilde{X}_i^*$ is a direct cause of $\tilde{X}_j^*$ and there is no other directed path from $\tilde{X}_i^*$ to $\tilde{X}_j^*$.  For those nodes which are not a direct-and-only-direct node of any leaf node, each of them forms a separate group. We call the set of all such groups ordered according to the causal ordering of the non-leaf nodes in DAG $\tilde{G}^*$ a recursive group decomposition of $\tilde{G}^*$, denoted by $\mathcal{G}_{\tilde{G}^*}$.
\end{Defi}

\paragraph{Example Set 3}
As seen from the process of recursive group decomposition, each non-leaf node is in one and only one recursive group, and it is possible for multiple leaf nodes to be in the same group. Therefore, in total there are $(n-l)$ recursive groups.  For example, for $\tilde{G}_A$ given in Figure~\ref{fig:illust_D}(a), a corresponding group structure for the corresponding $\tilde{G}^*$ is $\mathcal{G}_{\tilde{G}_A^*} = ( \{\tilde{X}_1^*\} \rightarrow \{\tilde{X}_2^*, \tilde{X}_5^*\}  \rightarrow \{\tilde{X}_3^*, \tilde{X}_6v\}  \rightarrow \{\tilde{X}_4^*, \tilde{X}_7^*, \tilde{X}_8^*\}       )$, and for $\tilde{G}_B$ in Figure~\ref{fig:illust_D}(b), there is only one group: $\mathcal{G}_{\tilde{G}_B^*} = ( \{\tilde{X}_1^*, \tilde{X}_2^*, \tilde{X}_3^*, \tilde{X}_4^*\})$. For both $\tilde{G}_C$ and $\tilde{G}_D$, given in Figure~\ref{fig:illust_D}(c), a recursive group decomposition is $ ( \{\tilde{X}_1^*\}  \rightarrow \{\tilde{X}_2^*, \tilde{X}_3^*\}  \rightarrow \{\tilde{X}_4^*\}  \rightarrow   \{\tilde{X}_5^*, \tilde{X}_6^*\} )$. 

Note that the causal ordering and the recursive group decomposition of given variables according to the graphical model $\tilde{G}^*$  may not be unique. For instance, if $\tilde{G}^*$ has only two variables $\tilde{X}_1^*$ and $\tilde{X}_2^*$ which are not adjacent, both decompositions $(\tilde{X}_1^*  \rightarrow \tilde{X}_2^*)$ and $(\tilde{X}_2^*  \rightarrow \tilde{X}_1^*)$ are correct.  Consider $\tilde{G}^*$ over three variables, $\tilde{X}_1^*, \tilde{X}_2^*, \tilde{X}_3^*$, where $\tilde{X}_1^*$ and $\tilde{X}_2^*$ are not adjacent and are both causes of $\tilde{X}_3^*$; then both $(\tilde{X}_1^*  \rightarrow \{ \tilde{X}_2^*, \tilde{X}_3^* \})$ and $(\tilde{X}_2^*  \rightarrow \{ \tilde{X}_1^*, \tilde{X}_3^* \})$ are valid recursive group decompositions.

We first present a procedure to construct the recursive group decomposition and the causal ordering among the groups  from the estimated $\mathbf{A}^{\mathrm{NL}}$.   We will further show that  the recovered recursive group decomposition is always asymptotically correct under assumption A4.

\subsubsection{Construction and Identifiability of Recursive Group Decomposition} \label{Eq:procedure}

First of all, Lemma~\ref{prop:recover_ica} tells us that $\hat{\mathbf{A}}^{\mathrm{NL}}$ in (\ref{Eq:ICA}) is identifiable up to permutation and scaling columns. Let us start with the asymptotic case, where the columns of the estimated $\mathbf{A}^{\mathrm{NL}}$  from values of ${X}_i$ are a permuted and rescaled version of the columns of  $\mathbf{A}^{\mathrm{NL}}$.   In what follows the permutation and rescaling of the columns of $\mathbf{A}^{\mathrm{NL}}$ does not change the result, so below we just work with the true $\mathbf{A}^{\mathrm{NL}}$, instead of its estimate.

$\tilde{X}_i^*$ and $\tilde{X}_i$ follow the same causal DAG, $\tilde{G}$, and $\tilde{X}_i^*$ are causally sufficient, although some variables among them (corresponding to leaf nodes in $\tilde{G}^*$) are determined by their direct causes.  Let us find the causal ordering of $\tilde{X}_i^*$.  If there are no deterministic relations and the values of $\tilde{X}_i^*$ are given, the causal ordering can be estimated by recursively performing regression and checking independence between the regression residual and the predictor~\citep{DirectLingam}. Specifically, if one regresses all the remaining variables on the root cause, the residuals are always independent from the predictor (the root cause). After detecting a root cause, the residuals of regressing all the other variables on the discovered root cause are still causally sufficient and follow a DAG. One can repeat the above procedure to find a new root cause over such regression residuals, until no variable is left.

However, in our case we have access to $\mathbf{A}^{\mathrm{NL}}$ but not the  values of $\tilde{X}^*_i$. Fortunately, the independence between regression residuals and the predictor can still be checked by analyzing $\mathbf{A}^{\mathrm{NL}}$. Recall that $\tilde{\mathbf{X}}^* = \mathbf{A}^{\mathrm{NL}} \tilde{\mathbf{E}}^{\mathrm{NL}}$, where the components of $\tilde{\mathbf{E}}^{\mathrm{NL}}$ are independent. Without loss of generality, here we assume that all components of $\tilde{\mathbf{E}}^{\mathrm{NL}}$ are standardized, i.e., they have a zero mean and unit variance. 
Denote by  $\mathbf{A}^{\mathrm{NL}}_{i\cdot}$ the $i$th row of $\mathbf{A}^{\mathrm{NL}}$. We have $\mathbb{E}[\tilde{X}^*_j\tilde{X}^*_i] = \mathbf{A}^{\mathrm{NL}}_{j\cdot} \mathbf{A}^{NL\intercal}_{i\cdot}$ and $\mathbb{E}[{\tilde{X}^{*2}_i}] = \mathbf{A}^{\mathrm{NL}}_{i\cdot} \mathbf{A}^{NL\intercal}_{i\cdot} = ||\mathbf{A}^{\mathrm{NL}}_{i\cdot}||^2$. The regression model for $\tilde{X}^*_j$ on $\tilde{X}^*_i$ is 
\begin{flalign} \nonumber
\tilde{X}^*_j = \frac{\mathbb{E}[\tilde{X}^*_j\tilde{X}^*_i]}{\mathbb{E}[\tilde{X}^{*2}_i]} \tilde{X}^*_i + R_{j\leftarrow i}  = \frac{\mathbf{A}^{\mathrm{NL}}_{j\cdot} \mathbf{A}^{NL\intercal}_{i\cdot}}{||\mathbf{A}^{\mathrm{NL}}_{i\cdot}||^2} \tilde{X}^*_i + R_{j\leftarrow i}.
\end{flalign}
Here the residual can be written as 
\begin{flalign} \nonumber
R_{j\leftarrow i} &=  \tilde{X}^*_j - \frac{\mathbf{A}^{\mathrm{NL}}_{j\cdot} \mathbf{A}^{NL\intercal}_{i\cdot}}{||\mathbf{A}^{\mathrm{NL}}_{i\cdot}||^2} \tilde{X}^*_i  \\  \label{Eq:res}
&= \underbrace{\big( \mathbf{A}^{\mathrm{NL}}_{j\cdot} - \frac{\mathbf{A}^{\mathrm{NL}}_{j\cdot} \mathbf{A}^{NL\intercal}_{i\cdot} \mathbf{A}^{\mathrm{NL}}_{i\cdot}}{||\mathbf{A}^{\mathrm{NL}}_{i\cdot}||^2}  \big)}_{
\triangleq \alpha_{j\leftarrow i}} \tilde{\mathbf{E}}^{\mathrm{NL}}.
\end{flalign}

If  for all $j$, $R_{j\leftarrow i}$ is either zero or independent from $\tilde{X}^*_i$, we consider $\tilde{X}^*_i$ as the current root cause and put it and all the other variables which are deterministically related to it in the first group, which is a root cause group. Now the problem is whether we can check for independence between nonzero residuals $R_{j\leftarrow i}$ and the predictor $\tilde{X}^*_i$. Interestingly, the answer is yes, as stated in the following proposition.
\begin{proposition} \label{Prop_ind_A}
Suppose assumption A4 holds. For variables $\tilde{\mathbf{X}}^*$ generated by (\ref{Eq:FA}), regression residual $R_{j\leftarrow i}$ given in (\ref{Eq:res}) is independent from variable $\tilde{{X}}_i^*$ if and only if
\begin{equation} \label{Eq:check_ind_asym}
| \alpha_{j\leftarrow i} \circ  \mathbf{A}^{\mathrm{NL}}_{i\cdot}|_1 = 0,
\end{equation}
where $\circ$ denotes entrywise product.
\end{proposition} 

So we can check for independence as if the values of $\tilde{\mathbf{X}}^*$ were given. Consequently, we can find the root cause group. 

We then consider the  residuals of regressing all the remaining variables $\tilde{X}^*_k$  on the discovered root cause as a new set of variables. Note that like the variables $\tilde{X}_j^*$, these variables are also linear mixtures  of $\tilde{{E}}_i$.  Repeating the above procedure on this new set of variance will give the second root cause and its recursive group. Applying this procedure repeatedly until no variable is left finally discovers all recursive groups following the causal ordering.
The constructed recurse group decomposition is asymptotically correct, as stated in the following proposition.

\begin{proposition} \label{non-gaussian:general}  {\bf (Identifiable recursive group decomposition)}
Let ${X}_i$ be generated by the CAMME with the corresponding measurement-error-free variables generated by the causal DAG $\tilde{G}$ and suppose assumptions \textcolor{black}{A0} and A4 hold. The recursive group decomposition constructed by the above procedure is asymptotically correct, in the sense that as the sample size $N\rightarrow \infty$, 
if non-leaf node $\tilde{X}_i$ is a cause of non-leaf node $\tilde{X}_j$, then the recursive group which $\tilde{X}_i$ is in precedes the group which $\tilde{X}_j$ belongs to. However, the causal ordering among the nodes within the same recursive group may not be identifiable.
\end{proposition}

The result of Proposition~\ref{non-gaussian:general} applies to any DAG structure in $\tilde{G}$.  Clearly, the indentifiability can be naturally improved if additional assumptions on the causal structure $\tilde{G}$ hold. In particular, to recover information of $\tilde{G}$, it is essential to answer the following questions.
\begin{itemize}
\item Can we determine which nodes in a recursive group are leaf nodes?
\item Can we find the causal edges into a particular node as well as their causal coefficients?
\end{itemize}

Below we will show that under rather mild assumptions, the answers to both questions are yes.



\subsubsection{Identifying Leaf Nodes and Individual Causal Edges}
If for each recursive group we can determine which variable is the non-leaf node, the causal ordering among the variables $\tilde{X}^*_i$ is then fully known. The causal structure in $\tilde{G}^*$ as well as the causal model can then be readily estimated by regression: for a leaf node, its direct causes are those non-leaf nodes that determine it; for a non-leaf node, we can regress it on all non-leaf nodes that precede it according to the causal ordering, and those predictors with non-zero linear coefficients are its parents. (Equivalently, its parents are the nodes that causal precede it and in its Markov blanket.)

Now the problem is whether it is possible to find out which variable in a given recursive group is a leaf node; if all leaf nodes are found, then the remaining one is the (only) non-leaf node.  We may find leaf nodes by ``looking backward" and ``looking forward"; the former makes use of the parents of the variables in the considered group, and the latter exploits the fact that leaf nodes do not have any child.


\begin{proposition} {\bf (Leaf node determination by ``looking backward")} 
Suppose the observed data were generated by the CAMME where assumptions A0 and A4 hold.\footnote{In this non-Gaussian case (implied by assumption A4), the result reported in this proposition may still hold if one avoids the non-deterministic faithfulness assumption and assumes a weaker condition; however, for simplicity of the proof we currently still assume non-deterministic faithfulness.} Let the sample size $N\rightarrow \infty$.  Then if assumption A5 holds, leaf node $O$ is correctly identified from values of $\mathbf{X}$ (more specifically, from the estimated $\mathbf{A}^{\mathrm{NL}}$ or the distribution of $\tilde{\mathbf{X}}^*$); alternatively,  
if assumption A6 holds, leaf nodes $O$ and $Q$ are correctly identified from values of $\mathbf{X}$. 
\begin{itemize}
\item[A5.] According to $\tilde{G}^*$, leaf node $O$ in the considered recursive group, $g^{(k)}$, has a parent which is not a parent of the non-leaf node in $g^{(k)}$.
\item[A6.] According to $\tilde{G}^*$, leaf nodes $O$ and $Q$ in the considered recursive group, $g^{(k)}$, are non-deterministically conditionally independent given some subset of the nodes in $g^{(1)}, g^{(2)}, ..., g^{(k)}$.
\end{itemize}
\label{prop_back}
\end{proposition}

\paragraph{Example Set 4}
Suppose assumptions A0 and A4 hold. 
\begin{itemize}
\item For $\tilde{G}_A$ in Figure~\ref{fig:illust_D}(a), assumption A6 holds for $\tilde{X}^*_7$ and $\tilde{X}^*_8$ in the recursive group $\{\tilde{X}^*_4, \tilde{X}^*_7, \tilde{X}^*_8\}$: they are non-deterministically conditionally independent given $\{\tilde{X}_2^*, \tilde{X}_4^*\}$; so both of them are identified to be leaf nodes from the estimated $\mathbf{A}^{\mathrm{NL}}$ or the distribution of $\tilde{\mathbf{X}}^*$, and $\tilde{X}^*_4$ can be determined as a non-leaf node. (In addition, assumption A5 holds for $\tilde{X}_8^*$, allowing us to identify this leaf node even if $\tilde{X}^*_7$ is absent in the graph.) 
\item For both $\tilde{G}_C$ and $\tilde{G}_D$ in Figure~\ref{fig:illust_D}(c), $\tilde{X}_6^*$, in the recursive group $\{\tilde{X}_5^*, \tilde{X}_6^*\}$, follows assumption A5 and can be found to be a leaf node from the distribution of $\tilde{X}^*_i$; accordingly, $\tilde{X}_5^*$ has to be a non-leaf node.
\item For $\tilde{G}_E$ in Figure~\ref{fig:illust_D}(d), assumption A5 holds for $\tilde{X}_5^*$ and $\tilde{X}_8^*$, which can then be found to be leaf nodes.
\end{itemize}


We can also determine leaf nodes by looking at the relationships between the considered variables  and the variables causally following them, as stated in the following proposition.
\begin{proposition} {\bf (Leaf node determination by ``looking forward")}
Suppose the observed data were generated by the CAMME where assumptions A0 and A4 hold. Then as the sample size $N\rightarrow \infty$, we can correctly identify the leaf node $U$ in the considered recursive group $g^{(k)}$ from values of $\mathbf{X}$ if assumption A7 holds for it:
\begin{itemize}
\item[A7.] For leaf node $U$ in $g^{(k)}$, there exists at least one node causally following $g^{(k)}$ that 1) is d-separated from $U$ by a subset of variables in  $g^{(1)}, ...,  g^{(k-1)},  g^{(k)}$ which does not include all parents of $U$ and 2) is a child of the non-leaf node in  $g^{(k)}$ .
\end{itemize}
\label{prop_forward}
\end{proposition}

\paragraph{Example Set 5}

Suppose assumptions A0 and A4 hold. 
\begin{itemize}
\item For data generated by $\tilde{G}_A$ in Figure~\ref{fig:illust_D}(a), we already found $\tilde{X}_4^*$ in recursive group $\{\tilde{X}^*_4, \tilde{X}^*_7, \tilde{X}^*_8\}$ to be a non-leaf node because of Proposition \ref{prop_back}. Proposition \ref{prop_forward} further indicates that $\tilde{X}^*_2$ (in group $\{\tilde{X}^*_2, \tilde{X}^*_5\}$) and $\tilde{X}^*_3$ (in group $\{\tilde{X}^*_3, \tilde{X}^*_6\}$) are non-leaf nodes, and all leaf nodes are identified. 
\item For  $\tilde{G}_B$ in Figure~\ref{fig:illust_D}(b), there is only one recursive group, and it does not provide further information by looking ``backward" or ``forward", and it is impossible to find the non-leaf node with Proposition \ref{prop_back} or \ref{prop_forward}.
\item For both $\tilde{G}_C$ and $\tilde{G}_D$ in Figure~\ref{fig:illust_D}(c), $\tilde{X}^*_6$ was found to be a leaf node due to Proposition \ref{prop_back}; thanks to Proposition~\ref{prop_forward}, the other leaf node, $\tilde{X}^*_3$, was also detected. In particular, in $\tilde{G}_C$, for leaf node $\tilde{X}^*_3$ both $\tilde{X}^*_4$ and $\tilde{X}^*_6$ satisfy the two conditions in assumption~\ref{prop_forward}; however, in $\tilde{G}_D$, for leaf node $\tilde{X}^*_3$ only $\tilde{X}^*_4$ satisfies them.    All leaf nodes were successfully found. 
\item For $\tilde{G}_E$ in Figure~\ref{fig:illust_D}(d), Proposition \ref{prop_back} already allows us to identify leaf nodes $\tilde{X}_5^*$ and $\tilde{X}_8^*$. Because assumption A7 holds for $\tilde{X}^*_4$ (for it $\tilde{X}_7^*$ satisfies the two conditions), we can further identify this leaf node.
\end{itemize}
For contaminated data generated by any of $\tilde{G}_A$, $\tilde{G}_C$, $\tilde{G}_D$, and $\tilde{G}_E$, now we can find all leaf nodes in the measurement-error-free causal model. 
One can then immediately estimate the whole measurement-error-free model, as seen next.


The above two propositions are about the identifiably of leaf nodes in the measurement-error-free causal model. By applying them to all leaf nodes, we have the (sufficient) conditions under which the causal graph of $\tilde{G}$ is fully identifiable.

\begin{proposition} {\bf (Full identifiability)}
Suppose the observed data were generated by the CAMME where assumptions A0 and A4 hold. Assume that for each leaf node in $\tilde{G}^*$, at least one of the three assumptions, A5, A6, and A7, holds. Then as the sample size $N\rightarrow \infty$, the causal structure in $\tilde{G}$ is fully identifiable from the contaminated observations.
\label{Prop:full}
\end{proposition}

In the general case, the causal structure in $\tilde{G}$ might not be fully identifiable, and the above propositions may allow partial identifiability of the underlying causal structure. Roughly speaking, the recursive group decomposition is identifiable in the non-Gaussian case; 
with Propositions \ref{prop_back} and \ref{prop_forward} one can further identify some leaf nodes as well as their parents.

\section{Conclusion and Discussions}


For variables of interest in various fields, including social sciences, neuroscience, and biology, the measured values are often contaminated by additional measurement error. Unfortunately, the outcome of existing causal discovery methods is sensitive to the existence of measurement error, and it is desirable to develop causal discovery methods that can estimate the causal model for the measurement-error-free variables without using much prior knowledge about the measurement error. To this end, this paper is concerned with the identifiability conditions for the underlying measurement-error-free causal model given contaminated observations. We have shown that under various conditions, the causal model is partially or even fully identifiable. 

Table~\ref{Summary_table} summarizes the identifiability results presented in this paper. Propositions \ref{coro1} and \ref{Prop:equivalence} make use of second-order statistics of the data, and Propositions~\ref{prop:recover_ica} to~\ref{Prop:full} further exploit non-Gaussianity of the data. 
The identifiability conditions reported in this paper are sufficient and might not be necessary. Below are some high-level interpretations of the conditions.
\begin{itemize}
\item Roughly speaking, in the Gaussian case (Propositions \ref{coro1} and \ref{Prop:equivalence}) the identifiability conditions are mainly about the number of leaf nodes in the measurement-error-free causal model and the sparsity of the underlying causal structure. The measurement-error-free causal model may be identifiable up to its equivalence class.
\item In the non-Gaussian case, the conditions for full identifiability are mainly about the sparsity of the measurement-error-free causal structure. 
\item In the non-Gaussian case, the identifiability conditions of the recursive group decomposition (including causal ordering between groups) are rather general (Proposition \ref{prop:recover_ica}). 
\item The identifiability of the measurement-error-free causal model may greatly benefit from additional knowledge about the measurement error (e.g, that all measurement errors have the same variance, as discussed in Propositions \ref{coro1} and~\ref{prop:recover_ica}).
\item Suppose assumptions A0 and A4 hold (the non-Gaussian case is considered); without additional knowledge about the measurement error, 
we conjecture that the necessary and sufficient condition for the non-leaf node to be identifiable is that at least one of the  three assumptions, A5, A6, and A7, holds. To falsify or prove this conjecture is part of our future work.
\end{itemize}

\begin{table*}[htp!]
 \centering
{
\centering
\caption{Summary of the identifiability results.}
  \begin{tabular}{|l || m{4cm} | m{7.2cm} |}
   \hline 
      Proposition $\#$     & Assumptions & What information of $\tilde{G}$ is identifiable? \\ \hline \hline
     Proposition \ref{coro1}   & A0, A1, and A2 & up to the equivalence class; leaf nodes identifiable \\ \hline 
          Proposition \ref{Prop:equivalence}   & A0, A1, and A3 & up to the equivalence class\\ \hline \hline
          
     Proposition~\ref{prop:recover_ica}   & A0, A4, A1, and A2  & Fully identifiable  \\ \hline
          Proposition~\ref{non-gaussian:general}   & A0 and A4 & Recursive group decomposition (including causal ordering between the groups) \\ \hline
                   Proposition~\ref{prop_back}   & A0, A4, and A5 {\it or}  A6 for some leaf nodes in $\tilde{G}^*$ & Recursive group decomposition; the leaf nodes \\ \hline
                   Proposition~\ref{prop_forward}   & A0, A4, and A7  for some leaf nodes in $\tilde{G}^*$ & Recursive group decomposition; the leaf nodes \\ \hline
                   Proposition~\ref{Prop:full}   & A0, A4, and A5 {\it or} A6 {\it or} A7 for each leaf node & Fully identifiable \\ \hline
   \end{tabular}
}
\label{Summary_table}
\end{table*}

We note that in principle, all assumptions except A0 (regarding the causal Markov condition and non-deterministic faithfulness assumption related to causal model $\tilde{G}^*$) are testable from the observed data. This suggests that it is possible to develop practical causal discovery methods to deal with measurement error that are able to produce reliable information at least in the asymptotic case. 



We have verified the validity of the algorithms briefly given in the paper on large samples, including $\texttt{FA+EquVar}$ and $\texttt{FA+DPC}$ outlined in Sections \ref{Sec:G_CAMME_same} and \ref{Sec:G_CAMME_general}, respectively, and the procedure to find recursive group decomposition given in Section \ref{Sec:NG_general}.  All of them are two-step methods: in the first step, the first two methods apply factor analysis on the data, and the last procedure applies over-complete independent component analysis; in the second step all the methods do causal discovery with subsequent analysis on the estimated information of the canonical representation of the CAMMA. These methods might not be statistically efficient for the purpose of causal discovery because of estimation errors in the first step. We are currently studying their behavior on finite samples and aim at developing statistically more efficient algorithms.  Such methods are also expected to be able to learn the optimal number of leaf nodes in the causal graph (in this paper we assume this number is given).

It is worth noting that various kinds of background knowledge of the causal model may further help improve the identifiability of the measurement-error-free causal model. For instance, if one knows that all causal coefficients are smaller than one in absolute value, then the measurement-error-free causal model in Figure~\ref{fig:illust_D}(b) is immediately identifiable from contaminated data.  Our future research also includes establishing identifiability conditions that allow cycles in the measurement-error-free causal model and developing efficient methods for particular cases where each measurement-error-free variable has multiple measured effects or multiplied measurement-error-free variables generate a single measured effect.









\begin{thebibliography}{15}
\providecommand{\natexlab}[1]{#1}
\providecommand{\url}[1]{\texttt{#1}}
\expandafter\ifx\csname urlstyle\endcsname\relax
  \providecommand{\doi}[1]{doi: #1}\else
  \providecommand{\doi}{doi: \begingroup \urlstyle{rm}\Url}\fi

\bibitem[Bekker \& {ten Berge}(1997)Bekker and {ten Berge}]{Bekker97}
Bekker, P.~A. and {ten Berge}, J.~M.~F.
\newblock Generic global indentification in factor analysis.
\newblock \emph{Linear Algebra and its Applications}, 264:\penalty0 255--263,
  1997.

\bibitem[Chickering(2002)]{chickering2002optimal}
Chickering, D.~M.
\newblock Optimal structure identification with greedy search.
\newblock \emph{Journal of machine learning research}, 3\penalty0
  (Nov):\penalty0 507--554, 2002.

\bibitem[Eriksson \& Koivunen(2004)Eriksson and Koivunen]{Eriksson04}
Eriksson, J. and Koivunen, V.
\newblock Identifiability, separability, and uniqueness of linear {ICA} models.
\newblock \emph{IEEE Signal Processing Letters}, 11\penalty0 (7):\penalty0
  601--604, 2004.

\bibitem[Everitt(1984)]{Everitt84}
Everitt, B.~S.
\newblock \emph{An introduction to latent variable models}.
\newblock London: Chapman and Hall, 1984.

\bibitem[Glymour(2007)]{Glymour07}
Glymour, C.
\newblock Learning the structure of deterministic systems.
\newblock In Gopnik, A. and Schulz, L. (eds.), \emph{Causal learning:
  psychology, philosophy and computation}. Oxford University Press, 2007.

\bibitem[Kagan et~al.(1973)Kagan, Linnik, and Rao]{Kagan73}
Kagan, A.~M., Linnik, Y.~V., and Rao, C.~R.
\newblock \emph{Characterization Problems in Mathematical Statistics}.
\newblock Wiley, New York, 1973.

\bibitem[Kummerfeld et~al.()Kummerfeld, Ramsey, Yang, Spirtes, and
  Scheines]{Kummerfeld14}
Kummerfeld, E., Ramsey, J., Yang, R., Spirtes, P., and Scheines, R.
\newblock Causal clustering for 2-factor measurement models.
\newblock In Calders, T., Esposito, F., H{\"u}llermeier, R., and Meo, R.
  (eds.), \emph{Proc. ECML PKDD}.

\bibitem[Luo(2006)]{Luo06}
Luo, W.
\newblock Learning bayesian networks in semi-deterministic systems.
\newblock In \emph{Proc. Canadian Conference on Artificial Intelligence}, pp.\
  230--241, 2006.

\bibitem[Pearl(2000)]{Pearl00}
Pearl, J.
\newblock \emph{Causality: Models, Reasoning, and Inference}.
\newblock Cambridge University Press, Cambridge, 2000.

\bibitem[Scheines \& Ramsey(2017)Scheines and Ramsey]{Scheines16}
Scheines, R. and Ramsey, J.
\newblock Measurement error and causal discovery.
\newblock In \emph{Proc. CEUR Workshop 2016}, pp.\  1--7, 2017.

\bibitem[Shapiro(1985)]{Shapiro85}
Shapiro, A.
\newblock Identifiability of factor analysis: Some results and open problems.
\newblock \emph{Linear Algebra and its Applications}, 70:\penalty0 1--7, 1985.

\bibitem[Shimizu et~al.(2006)Shimizu, Hoyer, Hyv{\"a}rinen, and
  Kerminen]{Shimizu06}
Shimizu, S., Hoyer, P.O., Hyv{\"a}rinen, A., and Kerminen, A.J.
\newblock A linear non-{Gaussian} acyclic model for causal discovery.
\newblock \emph{Journal of Machine Learning Research}, 7:\penalty0 2003--2030,
  2006.

\bibitem[Shimizu et~al.(2011)Shimizu, Inazumi, Sogawa, Hyv{\"a}rinen, Kawahara,
  Washio, Hoyer, and Bollen]{DirectLingam}
Shimizu, S., Inazumi, T., Sogawa, Y., Hyv{\"a}rinen, A., Kawahara, Y., Washio,
  T., Hoyer, P.~O., and Bollen, K.
\newblock Directlingam: Adirect method for learning a linear non-gaussian
  structural equation model.
\newblock \emph{Journal of Machine Learning Research}, pp.\  1225--1248, 2011.

\bibitem[Silva et~al.(2006)Silva, Scheines, Glymour, and Spirtes]{Silva06}
Silva, R., Scheines, R., Glymour, C., and Spirtes, P.
\newblock Learning the structure of linear latent variable models.
\newblock \emph{Journal of Machine Learning Research}, 7:\penalty0 191--246,
  2006.

\bibitem[Spirtes et~al.(2001)Spirtes, Glymour, and Scheines]{Spirtes00}
Spirtes, P., Glymour, C., and Scheines, R.
\newblock \emph{Causation, Prediction, and Search}.
\newblock MIT Press, Cambridge, MA, 2nd edition, 2001.

\end{thebibliography}

\newpage
\section*{Appendix: Proofs}

\subsection*{A.1: Proof of Proposition \ref{coro1}}
\begin{proof}
According to Proposition~\ref{iden_2nd}, the variances of $E_i^*$ are identifiable. If they are not identical, then their smallest value is the variance of measurement errors $E_i$.  If the estimated variance of $E^*_j$ is greater than the smallest value, 
according to the definition of $E^*_i$ in (\ref{newE}), $\tilde{X}_j$ must be a leaf node in $\tilde{G}$. There are in total three types of edges in $\tilde{G}$. 
\begin{enumerate}
\item First consider the edges between leaf nodes. According to the definition of leaf nodes, there is no edge between them. Since we know which nodes are leaf nodes, it is guaranteed that there is no edge between them.
\item Then consider the edges between non-leaf nodes. Remove all the leaf nodes, and we have the set of non-leaf nodes, which is causally sufficient. The subgraph of $\tilde{G}$ over this set of variables satisfies the causal Markov condition and the faithfulness assumption. Then by applying any consistent constraint-based causal discovery method, such as the PC algorithm, this subgraph is corrected identified up to its equivalence class as $N \rightarrow \infty$. That is, the edges between non-leaf nodes are correctly identified.
\item Finally consider the edges between leaf nodes and non-leaf nodes. Each leaf node $\tilde{X}_i^*$ is a deterministic, linear function of its direct causes, which are among non-leaf nodes. Denote by the set of the direct causes of $\tilde{X}_i^*$ by $S_i$.  Recall that the covariance matrix of non-leaf nodes is non-singular. As a consequence, $\tilde{X}_i^*$, as a linear combination of elements of $S_i$, cannot be represented as a deterministic, linear combination of other non-leaf nodes; otherwise some leaf node can be written as a deterministic, linear function of other non-leaf nodes, leading to a contradiction. As a result, $S_i$ is identifiable and, as a consequence, the edges between each leaf node and non-leaf nodes are identifiable.
\end{enumerate}
Therefore, $\tilde{G}$ is identifiable at least up to its equivalence class. Furthermore, we know that for all edges between leaf nodes and non-leaf nodes, the direction points to the leaf nodes. 
\end{proof}

\subsection*{A.2: Proof of Proposition \ref{Prop:equivalence}}
\begin{proof}
The deterministic relations among $\tilde{X}_i^*$ and their conditional independence/dependence relations can be seen from the covariance matrix of $\tilde{X}_i^*$. First, notice that the conditional independence relations detected by DPC are correct but may not be complete. Hence, asymptotically speaking, the edges removed by DPC do not exist in the true graph, i.e., DPC does not produce any missing edge.

We then note that there is no deterministic relation among non-leaf nodes and that two non-leaf nodes in $\tilde{G}$ which are not adjacent are d-separated by a proper subset of other non-leaf nodes; as a consequence, DPC will successfully detect the conditional independence relationships as well as the edges between non-leaf nodes. Therefore, the edges between all non-leaf nodes are identifiable.

Next, we shall consider whether DPC is able to correctly identify the edges between leaf nodes (which are actually not adjacent in $\tilde{G}$) and those between leaf nodes and non-leaf nodes. First consider the relationships between leaf nodes.
According to the former part of assumption A3, all paths between leaf nodes $\tilde{X}_j$ and $\tilde{X}_k$ that go through at most one of  $\tilde{X}_p$ and $\tilde{X}_q$ are blocked by $ (\mathrm{PA}(\tilde{X}_j) \setminus \tilde{X}_p) \cup (\mathrm{PA}(\tilde{X}_k) \setminus \tilde{X}_q)$; the  paths that go through both of $\tilde{X}_p$ and $\tilde{X}_q$ are blocked by $\mathbf{S}_1$. Therefore, all paths between $\tilde{X}_j$ and $\tilde{X}_k$ are blocked by 
$ (\mathrm{PA}(\tilde{X}_j) \setminus \tilde{X}_p) \cup (\mathrm{PA}(\tilde{X}_k) \setminus \tilde{X}_q) \cup \mathbf{S}_1$, which does not deterministically determine $\tilde{X}_j$ or $\tilde{X}_k$, so such an conditional independence relationship is detected by DPC. Hence each pair of leaf nodes in $\tilde{G}$, $\tilde{X}_j$ and $\tilde{X}_k$,  are not adjacent in the output of DPC.

Finally, we consider the d-separation relationships between each leaf node $\tilde{X}_j$ and its non-adjacent non-leaf node $\tilde{X}_i$. According to the latter part of assumption A3, all paths between $\tilde{X}_j$ and $\tilde{X}_i$ that do not go through $\tilde{X}_r$ are blocked by $\mathrm{PA}(\tilde{X}_j)\setminus \tilde{X}_r$, and the paths between them that go through $\tilde{X}_r$ are blocked by $\mathbf{S}_2$. Therefore, all paths between $\tilde{X}_j$ and $\tilde{X}_i$ are blocked by $(\mathrm{PA}(\tilde{X}_j)\setminus \tilde{X}_r)\cup \mathbf{S}_2$. That is, the outcome of DPC does not contain any extra edge between leaf nodes and non-leaf nodes.

Hence, the skeleton given by DPC is the same as the skeleton of $\tilde{G}$ under assumptions A0, A1, and A3.  Furthermore, according to Lemma 4 in~\citep{Luo06}, DPC correctly identifies all colliders in $\tilde{G}$. Therefore, under such assumptions $\tilde{G}$ is recovered up to its equivalence class.
\end{proof}

\subsection*{A.3: Proof of Proposition \ref{prop:recover_ica}}
\begin{proof}
Denoted by $\hat{\mathbf{A}}^{\mathrm{NL}}$ the estimate of $\mathbf{A}^{\mathrm{NL}}$. 
First note that according to Corollary~\ref{coro1}, leaf nodes in $\tilde{G}$ are identifiable. Then for each leaf node $\tilde{X}_l$, find the combination of the non-leaf nodes which determines $\tilde{X}_l$ in terms of $\hat{\mathbf{A}}^{\mathrm{NL}}$, which is achieved by finding the combination of the rows of  $\hat{\mathbf{A}}^{\mathrm{NL}}$ corresponding to non-leaf nodes to determine the $l$-th row of $\hat{\mathbf{A}}^{\mathrm{NL}}$. All nodes involved in this combination are direct causes of $\tilde{X}_l$. The solution in this step is unique because the rows of  $\hat{\mathbf{A}}^{\mathrm{NL}}$ corresponding to non-leaf nodes are linearly independent, as implied by the non-deterministic relations among $\tilde{X}_i$ (or non-zero variances of $\tilde{E}_i$). 

We have found all edges between leaf nodes and non-leaf nodes, and the remaining edges are those between non-leaf nodes. If we remove all leaf nodes and edges into them from $\tilde{G}$, we have the causal graph over non-leaf nodes and the graph is still acyclic, and the set of non-leaf nodes is causally sufficient. Denote by $\hat{\mathbf{A}}^{\mathrm{NL}}_s$ the matrix consisting of the  rows of  $\hat{\mathbf{A}}^{\mathrm{NL}}$ corresponding to non-leaf nodes. According to the identifiability of LiNGAM~\citep{Shimizu06}, the causal relations among non-leaf nodes are uniquely determined by $\hat{\mathbf{A}}^{\mathrm{NL}}_s$ or its inverse.
\end{proof}

\subsection*{A.4: Proof of Proposition \ref{Prop_ind_A}}
\begin{proof}
Both $R_{j\leftarrow i}$ and $\tilde{{X}}_i^*$ are linear mixtures of independent, non-Gaussian variables $\tilde{E}_i$. According to the Darmois-Skitovich theorem~\citep{Kagan73},  $R_{j\leftarrow i}$ and $\tilde{X}^*_i$ are statistically independent if and only if the for any $k$, at most one of the $k$th entries of their coefficient vectors, $\alpha_{j\leftarrow i}$ and $ \mathbf{A}^{\mathrm{NL}}_{i\cdot}$, is non-zero, which is equivalent to the condition (\ref{Eq:check_ind_asym}).
\end{proof}

\subsection*{A.5: Proof of Proposition \ref{non-gaussian:general}}
\begin{proof}
In the constructed recursive group decomposition, each group has one and only one non-leaf node. Just consider the non-leaf nodes in the recursive decomposition. Combining Lemma 1 in~\citep{DirectLingam} and Proposition~\ref{Prop_ind_A}, one can see that the discovered causal ordering among them must be correct. 
\end{proof}

\subsection*{A.6: Proof of Proposition \ref{prop_back}}
\begin{proof}
In each recursive group, there is a single non-leaf node, and all the others are leaf-nodes. Denote by $\tilde{X}^{*(k)}_q$ the $q$th node in the recursive group $g^{(k)}$.  
Denoted by $\tilde{X}^{*(k)}_{\mathrm{NL}}$ the only non-leaf node in the recursive group $g^{(k)}$.  Denote by $\mathrm{PA}(\tilde{X}^{*(k)}_{\mathrm{NL}})$ the set of direct causes of $\tilde{X}^{*(k)}_{\mathrm{NL}}$ in $\tilde{G}^*$. 

First consider assumption A5. Let us regress each variable $\tilde{X}^{*(k)}_q$ in this group on all variables $\tilde{X}^*_i$ in the first $(k-1)$ recursive groups; in this regression task, all predictors are causally earlier than $\tilde{X}^{*(k)}_q$ because of the identifiable causal ordering among the recursive groups.  Although the realizations of variables $\tilde{X}^*_i$ are unknown, such regression models can be estimated from the estimated matrix $\hat{\mathbf{A}}^{\mathrm{NL}}$, as done in Section~\ref{Eq:procedure}, or by analyzing the estimated covariance matrix of $\tilde{\mathbf{X}}^*$, which is $\hat{\mathbf{A}}^{\mathrm{NL}} \hat{\mathbf{A}}^{NL\intercal}$ (we have assumed that $\mathbb{V}ar (\tilde{E}^{\mathrm{NL}}_i) = 1$ without loss of generality). There are two possible cases to consider.
\begin{itemize}
\item[i)] For the non-leaf node $\tilde{X}^{*(k)}_{\mathrm{NL}}$ in the $k$th group, all predictors with non-zero coefficients are its direct causes, and their set is $\textrm{PA}(\tilde{X}^{*NL}_{(k)})$.
\item[ii)] Then consider a leaf node in this group, $\tilde{X}^{*(k)}_{q'}$, $\tilde{X}^{*(k)}_{q'} \neq \tilde{X}^{*(k)}_{\mathrm{NL}}$. Recall that when regressing $\tilde{X}^{*(k)}_{q'}$ on the variables in causally earlier recursive groups, $\tilde{X}^{*(k)}_{\mathrm{NL}}$ is not among the predictors because it is also in the $k$th group. First note that each node in $\tilde{X}^{*(k)}_{\mathrm{NL}}$ is always d-connected to $\tilde{X}^{*(k)}_{q'}$ given any variable set that does not include $\tilde{X}^{*(k)}_{\mathrm{NL}}$. As a consequence, in the regression model for $\tilde{X}^{*(k)}_{q'}$, all predictors in $\mathrm{PA}(\tilde{X}^{*(k)}_{\mathrm{NL}})$ have non-zero coefficients, so all predictors with non-zero coefficients form a superset of $\mathrm{PA}(\tilde{X}^{*(k)}_{\mathrm{NL}})$.  
Furthermore, under assumption A5, $O$ has at least a direct cause that is not in $\textrm{PA}(\tilde{X}^{*(k)}_{\mathrm{NL}})$. Therefore, in the regression model for $O$, the set of predictors with non-zero coefficients is a proper superset of $\textrm{PA}(\tilde{X}^{*NL}_{(k)})$ (the former has more elements). 
\end{itemize}
That is, when regressing variables $\tilde{X}^{*(k)}_q$ in the considered group on variables in earlier groups, the non-leaf node, as well as possibly some of the leaf nodes, always has a smaller number of predictors with non-zero coefficients, compared to the model for leaf node $O$. Hence we can determine $O$ as a leaf node.

Then consider assumption A6. According to assumption A0, A6 implies that $O$ and $Q$ are d-separated by a proper subset of the variables. Consequently, they are not adjacent in $\tilde{G}^*$. 
Bear in mind that the non-leaf node in $g^{(k)}$ is adjacent to all leaf nodes in the same group in $\tilde{G}^*$. Therefore both $O$ and $Q$ are leaf nodes.
\end{proof}

\subsection*{A.7: Proof of Proposition \ref{prop_forward}}
\begin{proof}
We note that the recursive group decomposition can be correctly identified from the values of $\mathbf{X}$ as $N \rightarrow \infty$, as implied by Proposition \ref{non-gaussian:general}.  
Denote by $W$ the non-leaf node in $g^{(k)}$. 
Let us first find a subset of the nodes causally following $g^{(k)}$ in which each node, denoted by $S$, is always non-deterministically dependent on at least one of the nodes in $g^{(k)}$ relative to $g^{(1)} \cup g^{(2)} \cup ... \cup g^{(k)} \cup {S}$. Denoted by $\mathbf{S}$ this set of nodes. If assumption A7 holds, $\mathbf{S}$ is non-empty. 

We then see that the non-leaf node $W$ is always non-deterministically dependent on every node in $\mathbf{S}$. Suppose this is not the case, i.e., there is $S \in \mathbf{S}$ which is non-deterministically independent from $W$ given a subset of $g^{(1)} \cup g^{(2)} \cup ... \cup g^{(k)}$. Denote by $\mathbf{R}_1$ this subset. 
If $\mathbf{R}_1$ contains any leaf nodes \textcolor{black}{in $\tilde{G}^*$}, let us remove those leaf nodes from  $\mathbf{R}_1$ and denote by  $\mathbf{R}'_1$ the resulting variable set. Further note that  
 $S$ and $W$ are still de-separated by  $\mathbf{R}'_1$. 
 Then $U'$, a leaf node in $g^{(k)}$, is always d-separated from $S$ given  $\mathbf{R}_1' \cup (\mathrm{PA}(U')\setminus W)$.  Since all nodes in $\mathbf{R}_1' \cup (\mathrm{PA}(U')\setminus W)$ are non-leaf nodes, $W$ can not be represented as their linear combination; thus $U'$ is not their deterministic function. Furthermore, $S$ is not a deterministic function of nodes in $\mathbf{R}_1' \cup (\mathrm{PA}(U')\setminus W)$ either; otherwise, according to the construction procedure of the recursive group decomposition, $S$ will belong to $g^{(1)} \cup g^{(2)} \cup ... \cup g^{(k)}$ because all elements of $\mathbf{R}_1' \cup (\mathrm{PA}(U')\setminus W)$ belong to it. Hence any leaf node $U'$ in $g^{(k)}$ will be non-deterministically independent from $S$ relative to $g^{(1)} \cup g^{(2)} \cup ... \cup g^{(k)} \cup {S}$, so $S$ is 
non-deterministically independent from every node in $g^{(k)}$ given a subset of $g^{(1)} \cup g^{(2)} \cup ... \cup g^{(k)}$. That is, 
 $S \notin \mathbf{S}$, leading to a contradiction.
 
 Next, we show that for  leaf node $U$ in $g^{(k)}$, there exists at least one element of $\mathbf{S}$ which is non-deterministically conditionally independent from $U$ given a subset of $g^{(1)} \cup g^{(2)} \cup ... \cup g^{(k)}$. Denote by $V$ one of the nodes that causally follow $g^{(k)}$ and satisfy the two conditions in assumption A7. Because of condition 2), $V \in \mathbf{S}$. Condition 1) states that $V$ and leaf node $U$ are d-separated by a subset of variables in  $g^{(1)}, ...,  g^{(k-1)},  g^{(k)}$ that does not include all parents of $U$.  
 Denote by $\mathbf{R}_2$ this variable set. If $\mathbf{R}_2$ contains any leaf nodes \textcolor{black}{in $\tilde{G}^*$}, remove them from  $\mathbf{R}_2$ and denote by  $\mathbf{R}_2'$ the resulting variable set. $V$ and $U$ are still de-separated by  $\mathbf{R}_2'$, but all elements of  $\mathbf{R}_2'$ are non-leaf nodes. Because all non-leaf nodes \textcolor{black}{in $\tilde{G}^*$} are linearly independent, the parents of $U$ that are not in $\mathbf{R}_2'$ can not be written as linear combinations of the elements of $\mathbf{R}_2'$. Therefore, $U$ is not a deterministic function of $\mathbf{R}_2'$. Moreover, $V$ is not a deterministic function of $\mathbf{R}_2'$ either, because otherwise $V$ will not in groups causally following $g^{(k)}$. This means that leaf node $U$ is non-deterministically independent from $V$, as an element of $\mathbf{S}$, given $\mathbf{R}_2'$. 

That is, we can distinguish between leaf node $U$ and the non-leaf node in the same recursive group by checking non-deterministic conditional independence relationships in $\tilde{X}_i^*$.

\end{proof}

\subsection*{A.8: Proof of Proposition \ref{Prop:full}}
\begin{proof}
First note that under the assumptions in the proposition, the recursive group decomposition is identifiable, and all leaf nodes are asymptotically identifiable. 
The causal ordering among the variables $\tilde{X}^*_i$ is then fully known. The causal graph $\tilde{G}$ as well as the causal model can then be readily estimated by regression: for a leaf node, its direct causes are those non-leaf nodes that determine it; for a non-leaf node, we can regress it on all non-leaf nodes that causally precede it according to the causal ordering, and those predictors with non-zero linear coefficients are its parents. 
\end{proof}

\end{document}